 \documentclass[journal]{new-aiaa} 
\usepackage[utf8]{inputenc}
\usepackage{graphicx}
\usepackage{caption}
\usepackage{subcaption}
\usepackage{amsmath}
\usepackage[version=4]{mhchem}
\usepackage{siunitx}
\usepackage{longtable,tabularx}
\usepackage{enumitem}
\usepackage{multirow}

\usepackage{amsthm}
\theoremstyle{plain}

\newtheorem{proposition}{Proposition}

\newtheorem{definition}{Definition}
\newtheorem{remark}{Remark}
\newtheorem{assumption}{Assumption}

\title{Quadrotor Guidance for Tasks along Elliptical Boundary}

\author{Ravi Khandelwal \footnote{Graduate Student, Department of Aerospace Engineering, ravikhandelw@iisc.ac.in (Corresponding Author).} and Ashwini Ratnoo \footnote{Professor, Department of Aerospace Engineering, ratnoo@iisc.ac.in, Associate Fellow AIAA. }}
\affil{Indian Institute of Science, Bengaluru, India, 560012}

\begin{document}

\maketitle

\begin{abstract}
This paper presents a guidance method that enables a quadrotor uncrewed aerial vehicle to perform inspection tasks along an elliptical boundary enclosing a structure of interest. The proposed guidance approach leverages a bifurcation theory-based framework that supports multiple inspection tasks, including convergence to a desired viewpoint, switching between viewpoints along the boundary, constant-speed elliptical circumnavigation, and circumnavigation with sectorial dwell-time constraint. Stability analysis of the proposed guidance dynamics establishes convergence of the generated reference trajectories to the desired behavior. The control parameters in the proposed method further allow the settling time of the reference trajectory to be specified for each task. The resulting reference velocity expressions are closed-form and algebraically simple, facilitating practical implementation. A qualitative comparison with existing guidance methods highlights the flexibility of the proposed approach in achieving different inspection tasks through a unified framework, while relying on simpler guidance computations. The proposed guidance method is validated through numerical simulations using a six-degree-of-freedom quadrotor model with model uncertainties, and a cascaded sliding-mode controller that tracks the generated reference trajectories under bounded disturbances.
\end{abstract}

\section{Introduction}
\lettrine{U}{ncrewed} Aerial Vehicles (UAVs), particularly quadrotors, are increasingly utilized for surveillance tasks because of their ability to access difficult-to-reach areas, reduce human error, enhance safety, and offer compact and cost-effective operation \cite{mathe2015vision}, \cite{jordan2018state}. Elliptical paths are widely adopted for perimeter patrolling \cite{paley2009perimeter}, convoy monitoring \cite{borkar2017vector}, and target enclosure \cite{yue2021elliptical} due to their flexibility in shape, smoothness, and ability to efficiently circumscribe asymmetric regions. These properties make elliptical boundaries a preferable choice for standoff surveillance and inspection missions. Within such a framework, the quadrotor may be required to reach, hover and switch between designated viewpoints along the elliptical boundary or to perform circumnavigation along the boundary for continuous inspection.

A wide range of path-planning strategies are proposed for UAV-based inspection tasks. Several of these strategies formulate the inspection mission as an optimization problem, where the UAV is required to traverse a set of predefined viewpoints. For instance, Ref.~\cite{wanninger2020towards} applies an ant colony optimization technique to obtain a trajectory of minimal length around a large structure. In Ref.~\cite{guerrero2013uav}, the inspection route is determined using a Zermelo traveling salesman framework to compute a time-optimal sequence connecting all target points. The work in Ref.~\cite{papadopoulos2013asymptotically} introduces a random inspection tree algorithm that yields the shortest trajectory covering all the viewpoints. Similarly, Ref.~\cite{bolourian2019path} proposes a hybrid approach that combines genetic algorithm with the A-star method to plan paths for inspecting inaccessible bridge surface cracks. It is noteworthy that optimization-based methods require relatively high computational effort for deducing the planning solution. An alternative approach is presented in Ref.~\cite{schorner2021uav}, wherein an artificial potential field guides the UAV through successive viewpoints for indoor structure inspection. The methods reported in Refs.~\cite{wanninger2020towards, guerrero2013uav,  papadopoulos2013asymptotically, bolourian2019path, schorner2021uav} primarily undertake a generic distribution of viewpoints and do not consider cases where viewpoints may follow a specific geometric locus or other spatial constraints. For example, when they are arranged along an elliptical path, a more sophisticated guidance framework is desired.

Elliptical circumnavigation provides a strategy in which a UAV follows an elliptical path which can facilitate continuous inspection. Various approaches are proposed to generate elliptical circumnavigation trajectories, including vector field methods \cite{borkar2017vector, lawrence2008lyapunov, de2017guidance, kapitanyuk2018gvf, shivam2022vector, chen2023vector}, geometry-based methods \cite{gandolfo2016stable, shao2023distance}, limit cycle based method \cite{cheng2010application}, and sliding mode control \cite{liang2017sliding, wang2021adaptive}. A notable limitation of elliptical circumnavigation is that it does not inherently accommodate scenarios in which the UAV is required to hover at specific viewpoints. The methods in Refs.~\cite{borkar2017vector, lawrence2008lyapunov, de2017guidance, kapitanyuk2018gvf, shivam2022vector, chen2023vector, gandolfo2016stable, shao2023distance, cheng2010application, liang2017sliding, wang2021adaptive} are also not applicable to scenarios in which the UAV is required to switch between designated viewpoints. Furthermore, in situations where closer observation is necessary through specific regions of the boundary, constant speed elliptical circumnavigation may be inadequate, thereby motivating an advanced sectorial dwell time constrained circumnavigation along the elliptical boundary.

Existing UAV path planning methods generally address viewpoint-based planning, switching between viewpoints, and elliptical circumnavigation as independent problems, each involving different computational and information requirements. Such isolated treatment limits the flexibility and autonomy of UAV operations. Therefore, a unified framework integrating these trajectory generation strategies is desirable. Motivated by this requirement, we adopt a bifurcation theory–based approach. Bifurcation refers to the changes in a system’s dynamics that occur when its parameters vary and cross a critical threshold \cite{kuznetsov1998elements}. This phenomenon leads to alterations in the qualitative characteristics of the system’s fixed points, including modifications in their number and stability properties \cite{strogartz1994nonlinear}. Consequently, bifurcation can be employed to design guidance laws that generate a variety of desired UAV trajectories. For instance, Ref.~\cite{2020anjaly} utilizes a supercritical pitchfork bifurcation in the look-angle dynamics of a UAV relative to a stationary target, enabling stand-off tracking, spiral search, and target prosecution. In Ref.~\cite{bennet2011autonomous}, pitchfork bifurcation is applied to artificial potential field-inspired velocity dynamics, allowing UAVs to form single-ring, double-ring, and cluster formations. Reference \cite{clark2017autonomous} employs Hopf bifurcation to achieve circumnavigation around a target using multiple UAVs. In another work Ref.~\cite{srinivasu2022standoff}, a two-parameter transcritical bifurcation is introduced for the line-of-sight distance kinematics in UAV engagement with a stationary target, which facilitates circumnavigation while achieving a prescribed settling time.  The authors in Ref.~\cite{khandelwal2026_jgcd} leverage infinite-period bifurcation in the line-of-sight kinematics to guide a quadrotor to a desired viewpoint, switching between viewpoints, and circumnavigation on a standoff circle. The methods presented in Refs.~\cite{2020anjaly, bennet2011autonomous, clark2017autonomous, srinivasu2022standoff, khandelwal2026_jgcd} are, however, restricted to circular paths.

In view of the above discussion, a unified trajectory generation framework is developed to address the limitations of existing distinct approaches. The main contributions of this work are summarized as follows:
\begin{enumerate}
\item An ellipse is formulated using a decoupled two-variable transformation, providing a simplified representation framework. A new saddle-node infinite-period bifurcation is introduced in the transformed two-dimensional space using a single bifurcation parameter. Furthermore, stability analysis identifies the stable fixed point, which is then selected to coincide with the desired viewpoint on the elliptical boundary.

\item The proposed framework establishes a unified mechanism for realizing multiple inspection tasks, including reaching a viewpoint, switching between viewpoints along the ellipse, and continuous elliptical circumnavigation, through variation in the bifurcation parameter.

\item Closed-form settling-time expressions are derived for the proposed dynamics of the transformed variables, enabling control parameters to be tuned to achieve the desired convergence behavior of quadrotor reference trajectories.

\item The framework further enables dwell-time constrained elliptical circumnavigation wherein the quadrotor spends a specified amount of time along a given sector of the boundary. The effectiveness of the overall method is demonstrated through simulations in which reference trajectories generated by the guidance method are tracked using a six-degree-of-freedom (6-DOF) quadrotor model that incorporates aerodynamic drag, model uncertainties, and Dryden-model wind gusts.
\end{enumerate}

Overall, this work presents a quadrotor guidance framework for inspection along an elliptical boundary that is both flexible and computationally efficient. Preliminary results of this work were reported in Ref.~\cite{khand_acc2025}, wherein only kinematic-model simulation results related to reaching a viewpoint and switching between viewpoints were presented. However, that study did not investigate multiple inspection tasks, settling-time analysis, simulations using a 6-DOF quadrotor model,  model uncertainties and external disturbances such as wind, or a qualitative comparison with existing elliptical circumnavigation methods. The present work addresses these limitations and further demonstrates the effectiveness of the proposed approach.

The remainder of the paper is structured as follows: Section \ref{sec:problem description} presents the problem description. Section \ref{sec:proposed method} details the proposed bifurcation-based guidance method along with its stability analysis. A detailed description of the desired trajectory solutions for inspection tasks is provided in Sec.~\ref{sec:inspection tasks}. Section \ref{sec: quadrotor model and control} describes the quadrotor dynamic model and the cascaded sliding mode control architecture. A qualitative comparison with existing elliptical circumnavigation methods is provided in Sec.~\ref{sec: qualitative comparison}. Numerical simulation results are presented in Sec.~\ref{sec:simulation results}, followed by concluding remarks along with directions for future research in Sec.~\ref{sec:conclusions}.

\section{Problem Formulation} \label{sec:problem description}
Consider a planar scenario in which a quadrotor is tasked with performing standoff inspection of an arbitrarily shaped structure, as illustrated in Fig. \ref{geometry_engagment}. In doing so, an elliptical boundary is employed that circumscribes the structure. The inspection may require observing the structure from specific viewpoints, such as A and B, indicated in Fig. \ref{geometry_engagment}. A fixed reference frame is defined, where the center of the ellipse is denoted by O, with semi-major and semi-minor axes represented by $a$ and $b$, respectively. The quadrotor is located at a distance $R$ from O, with line-of-sight angle $\theta$. Its engagement kinematics, moving at speed $v$ with heading angle $\alpha$, are given by
\begin{equation}
  \dot{x} = v \cos \alpha, \quad \dot{y} = v \sin \alpha
  \label{x_dot_y_dot_po}
\end{equation}
where $(x,y)$ denotes the instantaneous position of the quadrotor. The control input is defined as $\mathbf{u} = [\dot{x}, \dot{y}]^\top$, representing the quadrotor’s velocity components. The objective is to develop a guidance strategy that governs the quadrotor’s velocity components and generates a trajectory from an arbitrary initial position outside the elliptical boundary to a desired viewpoint, such as A, as depicted in Fig. \ref{geometry_engagment}, with a specified line-of-sight angle $\theta_d$. As additional guidance solutions, the proposed logic should enable switching from viewpoint A to another viewpoint (for example, B) flying along the elliptical boundary, and also offer continuous circumnavigation along the boundary.
\begin{figure}
      \centering
      \includegraphics[width=0.7\textwidth]{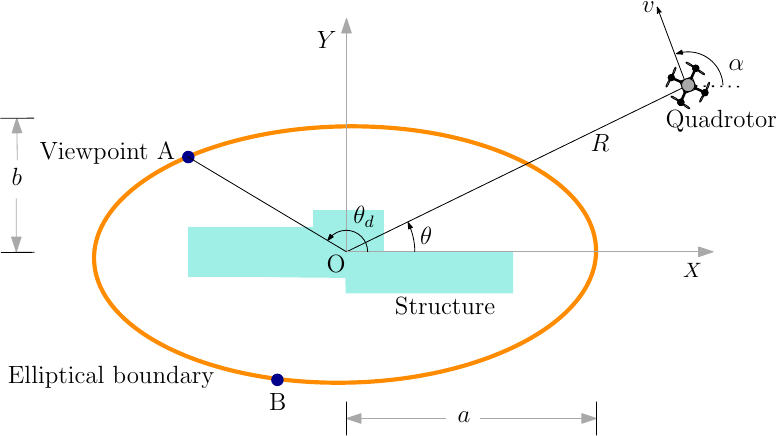}
      \caption{Illustration of the two-dimensional problem scenario.}
      \label{geometry_engagment}
\end{figure}   
The following assumptions are adopted in the present work.
\begin{assumption}
The geometry of the structure, along with the targeted viewpoints and regions of inspection, is assumed to be known a priori, and the corresponding inspection locations are projected onto a predefined elliptical boundary enclosing the structure.
\end{assumption}

\begin{assumption} 
The quadrotor is assumed to be equipped with either a gimballed camera or independent yaw control to maintain camera orientation toward the structure throughout the inspection. 
\end{assumption}

\section{Bifurcation-Based Guidance Method} \label{sec:proposed method}
This section first introduces a new two-variable representation of the quadrotor position. Utilizing which, a bifurcation based position dynamics is proposed which addresses the guidance objectives.

\subsection{Two-Variable Position Representation $(\beta, \phi)$}
The elliptical boundary, as illustrated in Fig. \ref{geometry_engagment}, can be defined as the zero-level set of an implicit function $f(x,y)$ for all $(x, y) \in \mathbb{R}^2$:
\begin{equation}
    S = \{ (x,y) \mid f(x,y) = \beta^2 - 1 = 0 \}
\end{equation}
where $\beta \in \mathbb{R^+}$ is a dimensionless elliptic radial variable defined as
\begin{equation}
    \beta := \sqrt{\frac{x^2}{a^2} + \frac{y^2}{b^2}}
    \label{beta_cov}
\end{equation}
Here, $\beta < 1$ and $\beta > 1$ correspond to $f(x,y) < 0$ and $f(x,y) > 0$, respectively, indicating points lying inside and outside the desired elliptical boundary $(\beta_d = 1)$. A point $(x, y) \in \mathbb{R}^2$ can equivalently be represented as
\begin{equation}
    x = a \beta \cos{\phi}, \quad y = b \beta \sin{\phi}
    \label{cov}
\end{equation}
where $\phi \in [0,2\pi)$ is the elliptic angle variable, computed using the
four-quadrant inverse tangent as
\begin{equation}
    \phi = \operatorname{atan2}\!\left(\frac{y}{b}, \frac{x}{a}\right)
    \label{phi_cov}
\end{equation}
Accordingly, the quadrotor’s position can alternatively be represented by $(\beta, \phi)$ using Eqs.~\eqref{beta_cov} and \eqref{phi_cov} via its Cartesian coordinates.
\begin{remark}
Unlike Cartesian and polar coordinates, a single variable $\beta \in \mathbb{R}^{+}$ uniquely defines an ellipse with centre O, semi-major axis $a\beta$, and semi-minor axis $b\beta$. Hence, $(\beta, \phi)$ provides a decoupled representation of the locus of all points on an ellipse, facilitating the proposed bifurcation-based guidance design.
\end{remark}

\begin{proposition}
The variables $(\beta_d, \phi_d)$ corresponding to the desired viewpoint at a line-of-sight angle $\theta_d$ are given by
\begin{equation}
   \beta_d = \sqrt{\frac{x_d^2}{a^2} + \frac{y_d^2}{b^2}}, \quad  
   \phi_d = \operatorname{atan2}\!\left(\frac{y_d}{b}, \frac{x_d}{a}\right)
   \label{main_phi_d}
\end{equation}
where \begin{equation}
x_d = R(\theta_d)\cos\theta_d, \quad 
y_d = R(\theta_d)\sin\theta_d
\label{des_cart_coord}
\end{equation} with
\begin{equation}
R(\theta_d)=
\frac{ab}{\sqrt{b^2\cos^2\theta_d+a^2\sin^2\theta_d}}
\label{gamma_def}
\end{equation}
\end{proposition}

\begin{proof}
The Cartesian coordinates of a point in $\mathbb{R}^2$ with respect to O can be expressed as
\begin{equation}
x=R\cos\theta,\qquad
y=R\sin\theta
\label{polar_sub}
\end{equation}
Using Eq.~\eqref{polar_sub} and the ellipse equation $\frac{x^2}{a^2} + \frac{y^2}{b^2} = 1$ results in
\begin{equation}
R^2\left(\frac{\cos^2\theta}{a^2}
+\frac{\sin^2\theta}{b^2}\right)=1
\label{eq:ellipse1}
\end{equation}
Rearranging Eq.~\eqref{eq:ellipse1} yields
\begin{equation}
R(\theta)=
\frac{ab}
{\sqrt{b^2\cos^2\theta+a^2\sin^2\theta}}
\label{gamma_intermediate}
\end{equation}
Using Eq.~\eqref{gamma_intermediate} in \eqref{polar_sub} with $\theta=\theta_d$ leads to the desired Cartesian coordinates given in Eq.~\eqref{des_cart_coord}. Further, the variables $(\beta_d, \phi_d)$ are obtained using Eq.~\eqref{des_cart_coord} in Eqs.~\eqref{beta_cov} and \eqref{phi_cov}, respectively. 
\end{proof}

\subsection{Proposed Dynamics}
Consider the proposed $(\beta, \phi)$ position dynamics of the quadrotor as
\begin{align}
    \dot{\beta} &= \lambda_{\beta}(e^{-\beta} - \rho)    \label{desired_beta_dot} \\
    \dot{\phi} &=  \lambda_{\phi}  (\mu -  k \sin{\phi})  \label{desired_phi_dot} 
\end{align}
where $\rho \in \mathbb{R}^{+}$, $\mu \in \mathbb{R}$, and $k \in \{-1,0,1\}$ denote the elliptic radial parameter, bifurcation parameter, and constant, respectively. The parameters $\lambda_{\beta}, \lambda_{\phi} \in \mathbb{R}^{+}$ represent the radial control parameter and angular control coefficient, respectively. The stability analysis of fixed points of Eqs.~\eqref{desired_beta_dot} and \eqref{desired_phi_dot} is discussed subsequently.   

\subsubsection{ Elliptic Radial Dynamics}
The fixed point $\beta^*$ for which $\dot{\beta}=0$ is evaluated from  Eq.~\eqref{desired_beta_dot} as 
\begin{equation}
    \beta^* = -\ln{\rho}
    \label{beta_star}
\end{equation}
To analyze the stability property, the sign of $\textrm{d}{\dot{\beta}}/\textrm{d}{\beta}$ is determined at this fixed point as
\begin{equation}
    {\frac{\textrm{d}{\dot{\beta}}}{\textrm{d}{\beta}}} \Bigr\rvert_{\beta^*} = - {\rho}\lambda_{\beta} < 0 
    \label{stability_of_beta}
\end{equation}
The sign of Eq.~\eqref{stability_of_beta} shows that $\beta^*$ is a stable fixed point. The corresponding vector field plot of the proposed $\dot{\beta}$ dynamics is shown in Fig. \ref{fig:beta_vector_field}. For reaching the desired elliptical boundary, we set $\beta^* = \beta_d =1$ and hence the elliptic radial parameter can be determined using Eq.~\eqref{beta_star} as
\begin{equation}
    \rho = e^{-\beta_d}
    \label{p_parameter}
\end{equation}
On integrating  Eq.~\eqref{desired_beta_dot}, it yields
\begin{equation}
     \beta(t) = \ln \left( \frac{1 - e^{-\rho {\lambda_\beta} t}(1 - \rho e^{\beta_0})   }{\rho} \right) 
         \label{beta_exp}
\end{equation}
where $\beta(t=0) = \beta_0$. Using Eq.~\eqref{p_parameter} in Eq.~\eqref{beta_exp}, it follows that as $t \to \infty$, $\beta(t) \to \beta_d$, that is, the quadrotor converges to the desired elliptical boundary. Further, the parameter $\lambda_{\beta}$ in Eq.~\eqref{desired_beta_dot} can be used to control the settling time for reaching the stable fixed point $\beta^*$, which is defined subsequently.
\begin{figure}
    \centering
    \includegraphics[width=0.5\linewidth]{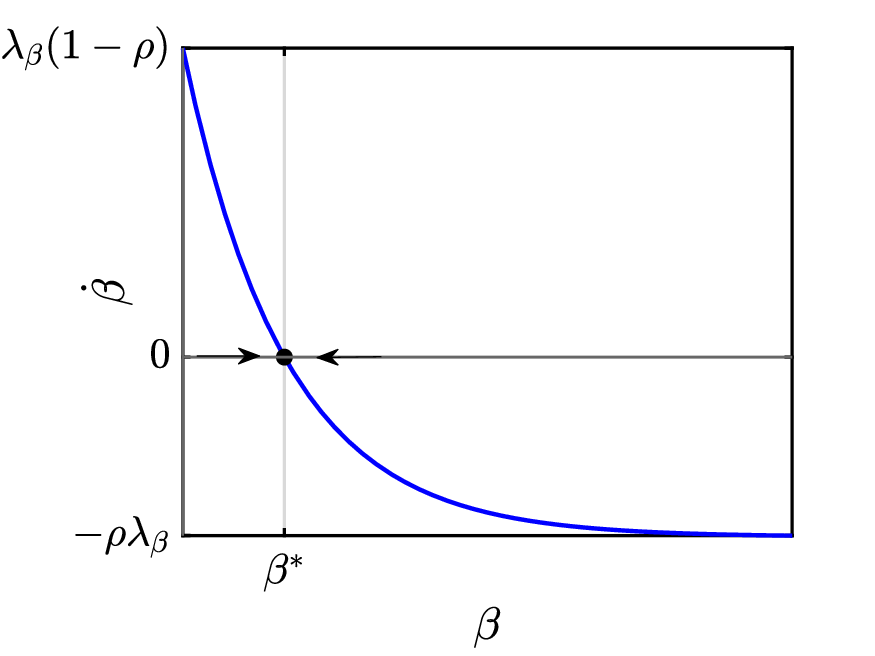}
    \caption{Vector field of the proposed elliptic radial dynamics.}
    \label{fig:beta_vector_field}
\end{figure}

\begin{definition}{(\cite{ogata2010modern})}
The \textit{settling time} $t_{s_\Gamma }$ of a system state $\Gamma (t)$ is defined as the time required for $\Gamma(t)$ to enter and remain within an $\epsilon_\Gamma$ neighborhood of a stable fixed point $\Gamma^*$, starting from the initial condition $\Gamma(0) = \Gamma_0$, that is, \label{definition_stat}
\begin{equation}
| \Gamma^* - \Gamma(t) | \leq \epsilon_\Gamma, \quad \forall t \geq t_{s_\Gamma}, \quad 0 < \epsilon_\Gamma \ll 1
\label{definition}
\end{equation}
\end{definition}

\begin{proposition}
   Given an initial value $\beta_0$ and a desired elliptical boundary $\beta_d$ corresponding to the fixed point $\beta^*$, the settling time $t_{s_\beta}$ for the dynamics in Eq.~\eqref{desired_beta_dot}, using Def.~\ref{definition_stat}, is given by
\begin{equation}
    t_{s_{\beta}} = -\frac{1}{\rho {\lambda_{\beta}}} \ln \left( \frac{1 - e^{{\epsilon_\beta}}}{1 - \rho e^{\beta_0}} \right),  \quad  \beta_0 > \beta_{d}
     \label{beta_dot_settling_time}
\end{equation}   
\end{proposition}

\begin{proof}
The result follows directly by using Eq.~\eqref{beta_exp} in Eq.~\eqref{definition}, wherein solving the equality yields the expression for $t_{s_{\beta}}$.    
\end{proof}
Using Eq.~\eqref{beta_dot_settling_time}, the radial control parameter $\lambda_{\beta}$ can be expressed in terms of the desired settling time as
\begin{equation}
    {\lambda_{\beta}}  = -\frac{1}{\rho { t_{s_{\beta}} }} \ln \left( \frac{1 - e^{{\epsilon_\beta}}}{1 - \rho e^{\beta_0}} \right)
    \label{lambda_beta_par}
\end{equation}
Here, $\rho$ in Eq.~\eqref{lambda_beta_par} is computed using Eq.~\eqref{p_parameter}.

\subsubsection{ Elliptic Angle Dynamics}
 Using Eq.~\eqref{desired_phi_dot} and evaluating $\dot{\phi} = 0$ when $ k \in \{-1, 1 \}$, the fixed points $\phi^*_1$ and $\phi^*_2$ are obtained as
\begin{align}
    \phi^*_1 =  \sin^{-1}\left( \frac{\mu}{k} \right), \quad  \phi^*_2 = \pi - \sin^{-1}\left( \frac{\mu}{k} \right) 
    \label{phi_fixed_points}
\end{align}
 The sign of  $\textrm{d}{\dot{\phi}}/\textrm{d}{\phi}$ is evaluated to analyze the stability of these fixed points using Eq.~\eqref{desired_phi_dot} as
\begin{align}
     {\frac{\textrm{d}{\dot{\phi}}}{\textrm{d}{\phi}}} \Bigr\rvert_{\phi^* } = -k{\lambda_{\phi}} \cos{\phi^*}
    \label{stability_of_phi}
\end{align}
Further analysis considers the two possible values of $k$.

\noindent\textit{2.1~ Stability of $\phi^*_1$ and $\phi^*_2$ for $k=1$}   \hfill\break
Using Eqs.~\eqref{phi_fixed_points} and \eqref{stability_of_phi} with $k=1$, 
\begin{align}
     {\frac{\textrm{d}{\dot{\phi}}}{\textrm{d}{\phi}}} \Bigr\rvert_{\phi^*_1 } = -{\lambda_{\phi}}\sqrt{1 - \mu^2}, \quad {\frac{\textrm{d}{\dot{\phi}}}{\textrm{d}{\phi}}} \Bigr\rvert_{\phi^*_2 } = {\lambda_{\phi}}\sqrt{1 - \mu^2}
    \label{stability_of_phi_k1}
\end{align}
For $\mu \in (-1,1)$, Eq.~\eqref{stability_of_phi_k1} shows that $\phi^*_1$ is a stable fixed point and it varies between $-\pi/2$ to $\pi/2$, that is, spanning quadrants first and fourth,  while $\phi^*_2$ is an unstable fixed point. Also, using Eq.~\eqref{phi_fixed_points}, at $\mu = \pm 1$, both fixed points merge, and they represent a half-stable fixed point, that is,
\begin{equation}
    \phi^* = \begin{cases}
        \frac{\pi}{2}, & \text{if } \mu = 1 \text{ and } k = 1 \\
        -\frac{\pi}{2}, & \text{if } \mu = -1 \text{ and } k = 1
    \end{cases}
    \label{k_1_90case}
\end{equation}
Figures \ref{vector_field_phi}(a) and (b) represent the half-stable fixed points for $k=1$. The associated bifurcation diagram is shown in Fig. \ref{bif_dia}(a). The bifurcation parameter is designed in such a way that a stable fixed point corresponds to the desired elliptic angle $\phi_d$, that is,  $\phi^*_1 = \phi_d$. Accordingly, using Eq.~\eqref{phi_fixed_points}, 
 \begin{align}
     \mu = k\sin{\phi_d}, \quad k = 1
     \label{mu_parameter_p1}
 \end{align}

\begin{figure}
	\centering
	\subfloat[{$\mu = -1$ and $k = 1$ }]
	{\includegraphics[width=0.45\textwidth]{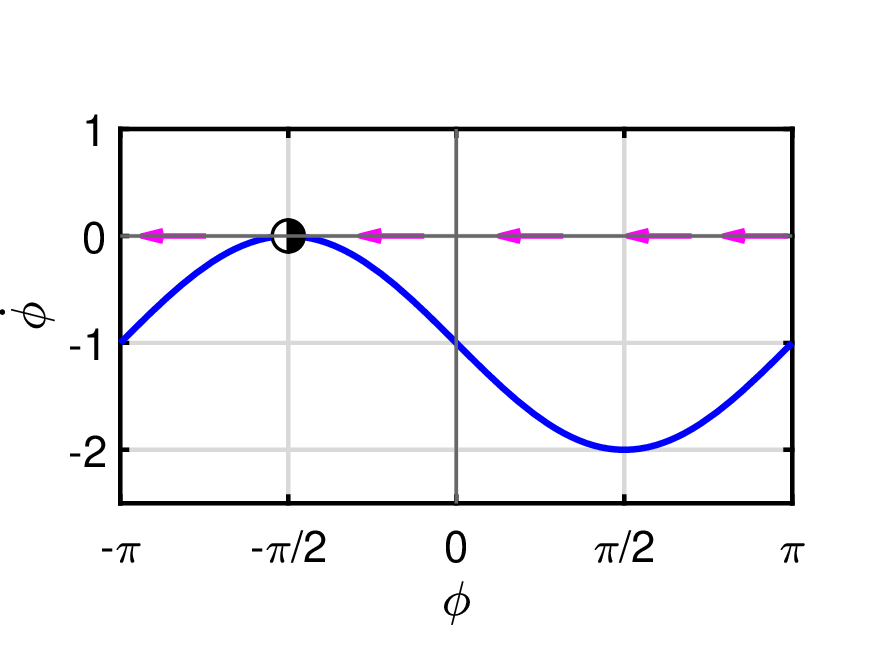}}
		\subfloat[{$\mu = 1$ and $k = 1$ }]
	{\includegraphics[width=0.45\textwidth]{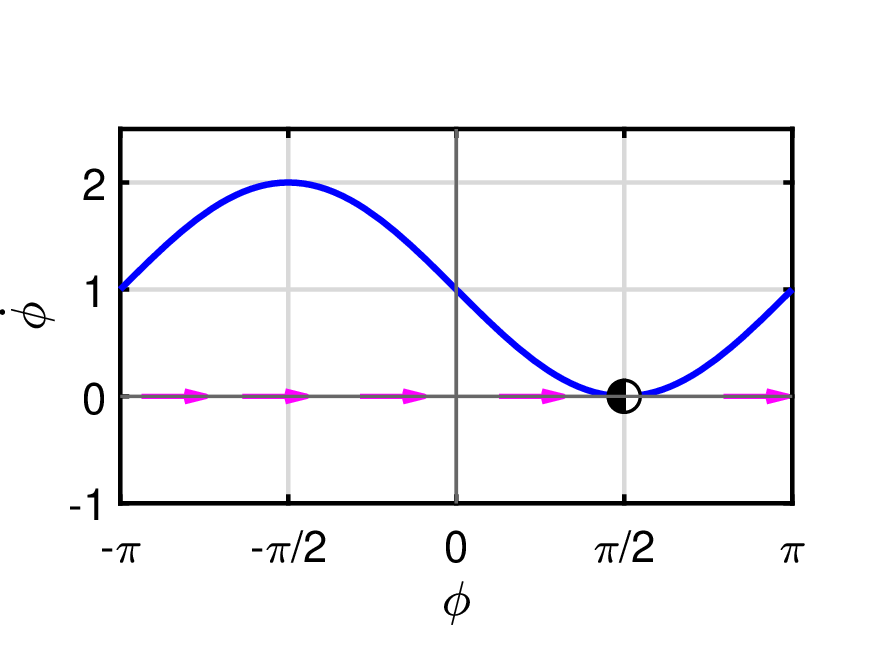}}\\
	\subfloat[{$\mu = -1$ and $k = -1$ }]
	{\includegraphics[width=0.45\textwidth]{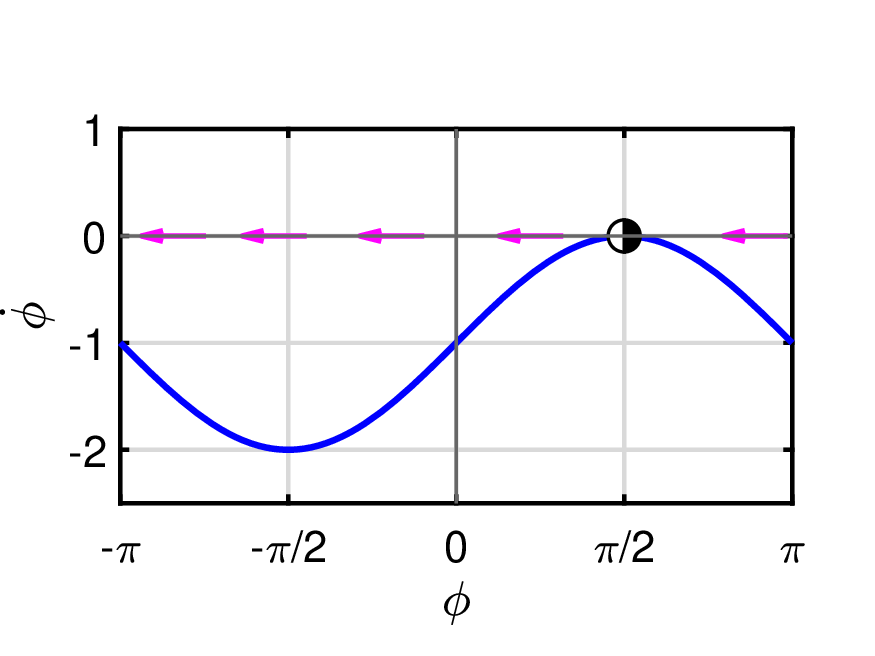}}
		\subfloat[{$\mu = 1$ and $k = -1$ }]
	{\includegraphics[width=0.45\textwidth]{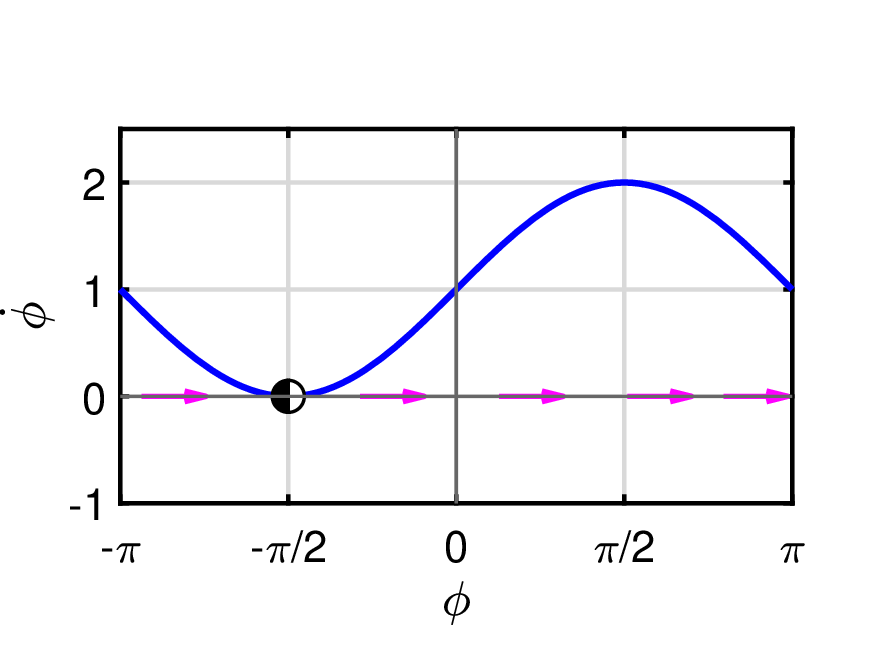}}
	\caption{Vector field for elliptic angle dynamics.}
	 \label{vector_field_phi}
\end{figure}

\begin{figure}
    \centering
    \subfloat[\centering $k = 1$]{{\includegraphics[width=0.5\textwidth]{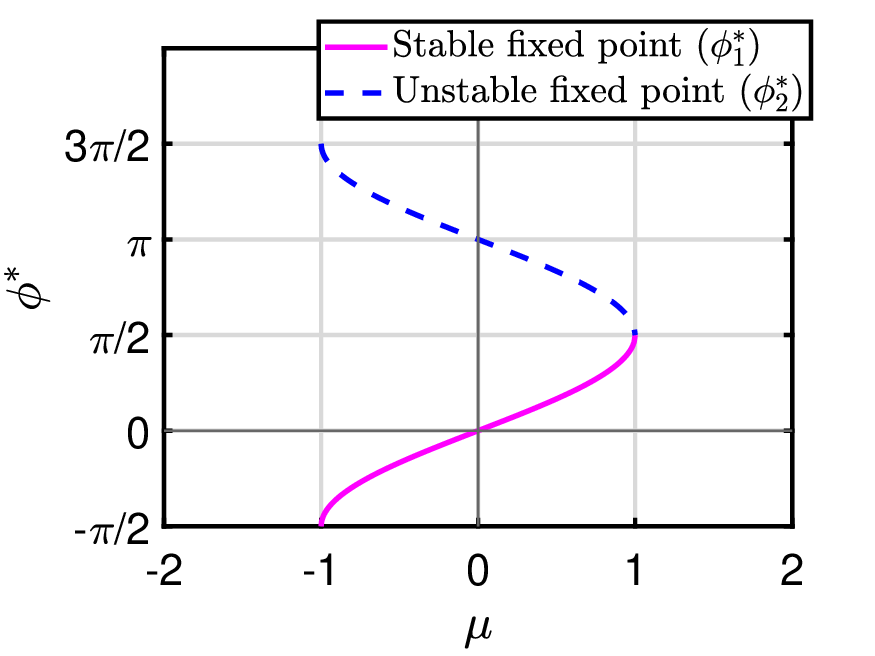} }}
    \subfloat[\centering $k = -1$ ]{{\includegraphics[width=0.5\textwidth]{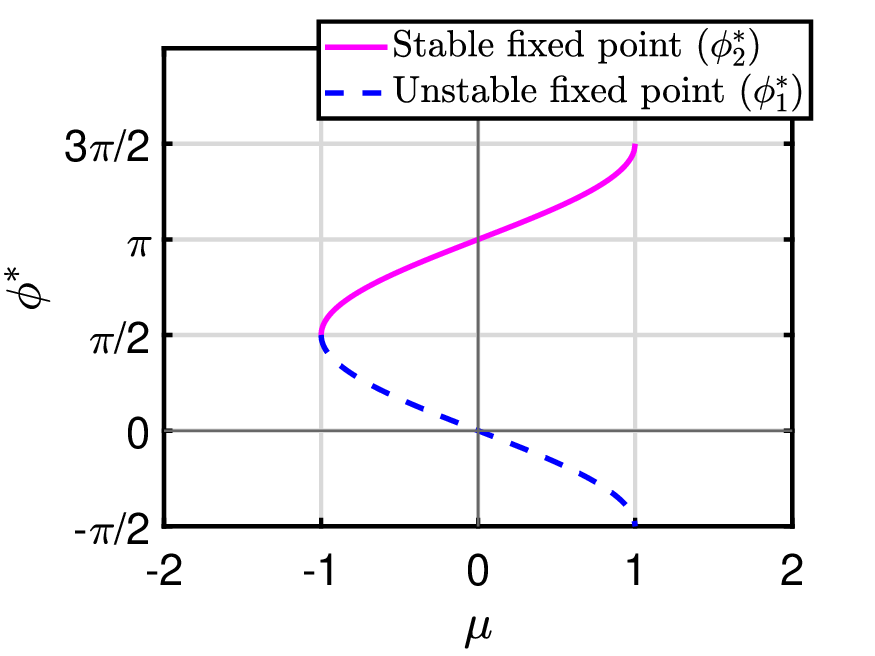} }}
    \caption{ Bifurcation diagram.}
    \label{bif_dia}
\end{figure}

\noindent\textit{2.2~ Stability of $\phi^*_1$ and $\phi^*_2$ for $k= -1$}   \hfill\break
Similar to case 2.1, using Eq.~\eqref{phi_fixed_points} in Eq.~\eqref{stability_of_phi} with $k=-1$ yields
\begin{align}
     {\frac{\textrm{d}{\dot{\phi}}}{\textrm{d}{\phi}}} \Bigr\rvert_{\phi^*_1 } = {\lambda_{\phi}}\sqrt{1 - \mu^2}, \quad {\frac{\textrm{d}{\dot{\phi}}}{\textrm{d}{\phi}}} \Bigr\rvert_{\phi^*_2 } = -{\lambda_{\phi}}\sqrt{1 - \mu^2}
    \label{stability_of_phi_kM1}
\end{align}
Equation \eqref{stability_of_phi_kM1} shows that for  $\mu \in (-1,1)$,  $\phi^*_1$ is an unstable fixed point while $\phi^*_2$ is a stable fixed point that varies from $\pi/2$ to $3\pi/2$,that is, quadrants second and third.  A half-stable fixed point is obtained at $\mu = \pm 1$ using Eq.~\eqref{phi_fixed_points} as
\begin{equation}
    \phi^* = \begin{cases}
        -\frac{\pi}{2}, & \textrm{if } \mu = 1 \textrm{ and } k = -1 \\
        \frac{\pi}{2}, & \textrm{if } \mu = -1 \textrm{ and } k = -1
    \end{cases}
    \label{k_m_1_90case}
\end{equation}
The corresponding half-stable fixed points for $k=-1$ are depicted in Figs.~\ref{vector_field_phi}(c) and (d). The associated bifurcation diagram is shown in Fig.~\ref{bif_dia}(b). To reach the desired elliptic angle, we determine the bifurcation parameter using Eq.~\eqref{phi_fixed_points} as $\phi_d = \phi^*_2 $, which leads to 
\begin{align}
    \mu = k\sin{\phi_d}, \quad k = -1
    \label{mu_parameter_m1}
\end{align}
The quadrant wise values of $k$ and $\mu$ are summarized in Fig. \ref{fig:bif_par_quad_wise}.
\begin{figure}
    \centering
    \includegraphics[width=0.45\linewidth]{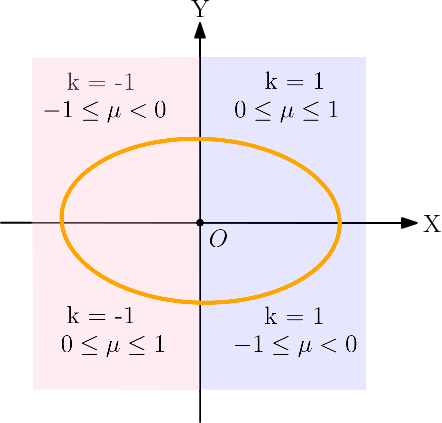}
    \caption{Values of $k$ and $\mu$ corresponding to stable fixed points of $\dot{\phi}$ located in specific quadrants.}
    \label{fig:bif_par_quad_wise}
\end{figure}
On integrating Eq.~\eqref{desired_phi_dot} when $ k \in \{-1, 1 \}$, it yields 
\begin{equation}
\phi(t)=
\begin{cases}
2\tan^{-1}\!\left(
\frac{m_1\eta(t)-c_1}
{\eta(t)-c_2}
\right),
& \text{if } \mu\in[-1,1]\setminus\{0\}\\[1ex]
2\tan^{-1}\!\left(
\tan\!\left(\frac{\phi_0}{2}\right)
e^{-\lambda_{\phi}kt}
\right),
& \text{if } \mu=0
\end{cases}
\label{phi(t) expression}
\end{equation}
where $\eta(t)=e^{-\lambda_{\phi}t\sqrt{k^{2}-\mu^{2}}}$,
$c_1=\frac{m_2m_4}{m_3}$,
$c_2=\frac{m_4}{m_3}$,
$m_1=\frac{k}{\mu}+\sqrt{\frac{k^2}{\mu^2}-1}$,
$m_2=\frac{k}{\mu}-\sqrt{\frac{k^2}{\mu^2}-1}$,
$m_3=\tan\!\left(\frac{\phi_0}{2}\right)-m_2$,
$m_4=\tan\!\left(\frac{\phi_0}{2}\right)-m_1$,
and $\phi(0)=\phi_0$. Using basic trigonometric identities, Eq.~\eqref{phi(t) expression} implies that $\lim_{t \to \infty} \phi(t) = \phi_d$ except when $\phi_0$ coincides with the unstable fixed point.

\begin{proposition}\label{prop3}
For the dynamics given in Eq.~\eqref{desired_phi_dot} and a given initial elliptic angle $\phi_0$, the settling time $t_{s_{\phi}}$ required for $\phi(t)$ to enter and remain within an $\epsilon_{\phi}$ neighborhood of the stable fixed point $\phi^*$ is given as follows.

If $\mu \in (-1,1)\setminus\{0\}$ and $\phi_0 \in \mathbb{S}_\pm$, then
{\small
\begin{equation}
\begin{aligned}
t_{s_\phi}
&= - \frac{1}{\lambda_\phi \sqrt{k^2-\mu^2}}
\ln \left(
\frac{ m_4 \left[ \tan \left(\tan^{-1}(m_2) \mp \frac{\epsilon_\phi}{2} \right) - m_2 \right] }
     { m_3 \left[ \tan \left(\tan^{-1}(m_2) \mp \frac{\epsilon_\phi}{2} \right) - m_1 \right] }
\right)
\end{aligned}
\label{phi_settling_a}
\end{equation}
}
elseif $\mu = 0$ and $\phi_0 \in \mathbb{S}_\pm$, then
\begin{equation} \label{phi_settling_b}
\begin{aligned}
    t_{s_\phi} &= - \frac{1}{\lambda_\phi k} 
    \ln \left( \frac{\tan((\phi^* \mp \epsilon_\phi)/2)}{\tan(\phi_0/2)} \right)
\end{aligned}
\end{equation}
Here the sets $\mathbb{S}_{+} := \{\phi_0 \mid \dot{\phi}(0) > 0\}$ and 
$\mathbb{S}_{-} := \{\phi_0 \mid \dot{\phi}(0) < 0\}$, and using Eqs.~\eqref{mu_parameter_p1} and \eqref{mu_parameter_m1}, the stable fixed point is  $\phi^{*} = 0$ for $k=1$ and $\phi^{*} = \pi$ for $k=-1$.
\end{proposition}

\begin{proof}
    See Appendix.  
\end{proof}
Note that angular control coefficient $\lambda_{\phi}$ can be explicitly expressed in terms of desired settling time $t_{s_\phi}$ using Eq.~\eqref{phi_settling_a} or \eqref{phi_settling_b}.

\begin{remark}
  For $|\mu| > 1$, the elliptic angle dynamics described by Eq.~\eqref{desired_phi_dot} possesses no fixed points. In the case $k = 1$, as $\mu$ varies from $-\infty$ to $\infty$, a half-stable fixed point appears at $\phi = -\pi/2$ when $\mu = -1$. This fixed point subsequently bifurcates into a pair of fixed points one stable and one unstable within the interval $\mu \in (-1, 1)$. At $\mu = 1$, these fixed points coalesce into another half-stable fixed point at $\phi = \pi/2$, as depicted in Fig. \ref{bif_dia}(a). An analogous reverse behavior occurs for $k = -1$, as shown in Fig. \ref{bif_dia}(b). Consequently, the dynamics defined by Eqs.~\eqref{desired_beta_dot} and \eqref{desired_phi_dot} exhibit saddle-node infinite period bifurcation \cite{strogartz1994nonlinear}.
\end{remark}

\section{Inspection Tasks}\label{sec:inspection tasks}
Using the analytical results derived in Sec.~\ref{sec:proposed method}, this section formulates guidance logic for four distinct inspection tasks along the boundary.

\subsection{Reaching a Viewpoint} \label{reach_vp}
In this scenario, the quadrotor is required to arrive and hover at a desired viewpoint, which is defined by the line-of-sight angle $\theta_d$, on the desired elliptical boundary ($\beta_d = 1$), starting from an initial position outside the elliptical boundary within a specified time $T_{vp}$. For this objective, the elliptic radial parameter $\rho$ is directly obtained using Eq.~\eqref{p_parameter}. The desired elliptic angle $\phi_d$ corresponding to $\theta_d$ is computed using Eq.~\eqref{main_phi_d}, while the radial control parameter $\lambda_{\beta}$ and angular control coefficient $\lambda_{\phi}$ are computed using $T_{vp}$ as $t_{s_{\beta}}$ in Eq.~\eqref{lambda_beta_par} and $T_{vp}$ as $t_{s_{\phi}}$ in Eq.~\eqref{phi_settling_a} or \eqref{phi_settling_b}, respectively. Further, depending on the quadrant in which $\phi_d$ lies, the bifurcation parameter $\mu$ is computed using Eqs.~\eqref{k_1_90case}, \eqref{mu_parameter_p1}, \eqref{k_m_1_90case}, and \eqref{mu_parameter_m1} as
\begin{equation}
\begin{aligned}
\mu &= k\,\sin \phi_d, \;\text{where}\;
k =
\begin{cases}
1, & \text{for } \phi_d \in \left(-\frac{\pi}{2},\,\frac{\pi}{2}\right) \\
-1, & \text{for } \phi_d \in \left(\frac{\pi}{2},\,\frac{3\pi}{2}\right) \\
\pm 1, & \text{for } \phi_d \in \left\{-\frac{\pi}{2},\,\frac{\pi}{2}\right\}
\end{cases}
\end{aligned}
\label{mu_parameter_for_all_quadrants}
\end{equation}

\subsection{Switching Between Viewpoints} \label{switch_vp}
Consider that the quadrotor is initially located at a viewpoint, say A, as shown in Fig.~\ref{geometry_engagment}, arriving there using the method described in Sec.~\ref{reach_vp}. It is desired that the quadrotor now flies along the elliptical boundary to another viewpoint (for example point B in Fig.\ref{geometry_engagment}) specified by, say $\theta_{d_B}$, in a desired time duration $T_{sw}$. Since the quadrotor's initial position is on the boundary, the elliptic radial parameter $\rho= e^{-\beta_d}$  ensures that it remains on it. The elliptic angle $\phi_{d_B}$ corresponding to $\theta_{d_B}$ is obtained using Eq.~\eqref{main_phi_d}. The required bifurcation parameter $\mu$ is computed using $\phi_d = \phi_{d_B}$ in Eq.~\eqref{mu_parameter_for_all_quadrants}. The angular control coefficient $\lambda_{\phi}$ is obtained using $\mu$ and $t_{s_{\phi}}= T_{sw}$ in Eq.~\eqref{phi_settling_a} or \eqref{phi_settling_b}. 

\subsection{Dwell Time Constrained Elliptical Circumnavigation}
This task enables continuous inspection along the elliptical boundary by incorporating observation requirements for specific regions of the structure. For example, the quadrotor is required to dwell for a duration $T_{dw}$ within the sector defined by $\theta_1 = \theta_d - \delta$ and $\theta_2 = \theta_d + \delta$, centered at the desired viewpoint $\theta_d$. The problem can be defined as $\int_{\theta(t) \in [\theta_d - \delta,\, \theta_d + \delta]} \! \textrm{d}t = T_{dw} $. The proposed guidance structure is inherently amenable to such a requirement. In Eq.~\eqref{desired_phi_dot}, it follows that for  $ |\mu|>1 $ there are no fixed points of $\dot{\phi}$, implying that $\phi$ varies continuously resulting in the emergence of a stable limit cycle. Figure \ref{VF_k1_MU1.5}(a) illustrates the counter-clockwise stable limit cycle corresponding to $\mu > 1$, while Fig. \ref{VF_k1_MU1.5}(b) depicts the clockwise limit cycle for $\mu < -1$ with $k =1$.
\begin{figure}
    \centering
    \subfloat[\centering $\mu = 1.1$ ]{{\includegraphics[width=0.37\textwidth]{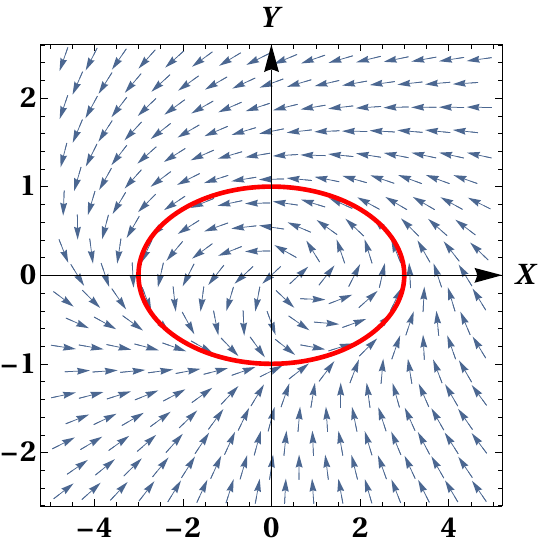} }}
    \subfloat[\centering $\mu = -1.1$]{{\includegraphics[width=0.37\textwidth]{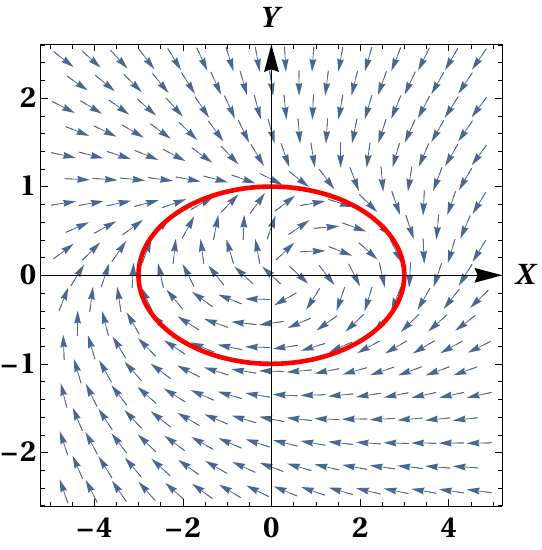} }}
    \caption{ Phase portrait of the proposed dynamics for $k =1$ when $|\mu| > 1$.}
    \label{VF_k1_MU1.5}
\end{figure}
To achieve the dwell-time feature, the bifurcation parameter $\mu$ is computed using Eq.~\eqref{desired_phi_dot} as
\begin{equation}
    f(\mu) = \lambda_{\phi}\,T_{dw} - \int_{\phi_1}^{\phi_2} 
    \frac{\mathrm{d}\phi}{\mu - k \sin\phi} = 0,\quad \textrm{where} \quad |\mu|>1
    \label{mu_dwell}
\end{equation}
wherein, $\phi_1$ and $\phi_2$ are obtained using the sector boundaries $\theta_1$ and $\theta_2$ in Eq.~\eqref{main_phi_d}, respectively. Equation \eqref{mu_dwell} is a nonlinear transcendental equation in the scalar unknown $\mu$ and cannot be solved in closed form. Accordingly, $\mu$ can be determined numerically. In that context, the existence and uniqueness of solution $\mu$ of Eq.~\eqref{mu_dwell} is important and it is established in Proposition~\ref{mu_existance}.

\begin{proposition}\label{mu_existance}
For given $\lambda_{\phi}>0$, $k>0$, $T_{dw}>0$, and
constants $\phi_2>\phi_1$, Eq.~\eqref{mu_dwell} has exactly
one feasible solution $\mu^\ast$, which satisfies
$\mu^\ast>kM$, where $M:=\max_{\phi\in[\phi_1,\phi_2]}\sin\phi.$
\end{proposition}

\begin{proof}
 We examine the continuity and monotonicity of the associated integral function in Eq.~\eqref{mu_dwell} over its feasible domain and analyze its limiting behavior at the domain boundaries.
Let $I:=[\phi_1,\phi_2]$, $\Delta:=\phi_2-\phi_1>0$,
$c:=\lambda_{\phi}T_{dw}>0$, and
\[
m:=\min_{\phi\in I}\sin\phi,\qquad
M:=\max_{\phi\in I}\sin\phi
\]
Consider
\begin{equation}
G(\mu):=\int_{\phi_1}^{\phi_2}
\frac{\mathrm{d}\phi}{\mu-k\sin\phi}
\label{eq:def_G}
\end{equation}
Using Eqs.~\eqref{mu_dwell} and \eqref{eq:def_G}, $f(\mu)=0$ is equivalent to $G(\mu)=c$. The integral is well-defined as an ordinary proper integral on
$\mathcal D=(-\infty,km)\cup(kM,\infty)$. For $\mu\in[km,kM]$, there exists $\phi_0\in I$ with $\mu=k\sin\phi_0$, so the integrand is singular and $\mu$ is infeasible. For $\mu<km$, we have $\mu-k\sin\phi\leq \mu-km<0, \quad \forall\,\phi\in I,$ which implies $G(\mu)<0<c$. Therefore, no solution exists in $(-\infty,km)$.

\textit{Continuity and strict monotonicity:}
Let $U:=(kM,\infty)$. For every $\mu\in U$ and $\phi\in I$, $\mu-k\sin\phi\geq\mu-kM>0.$ Thus, the integrand in Eq.~\eqref{eq:def_G} is continuous on $I$, and $G$ is continuous on $U$. Differentiating using Leibniz integral rule yields
\[
G'(\mu)=
-\int_{\phi_1}^{\phi_2}
\frac{\mathrm{d}\phi}{(\mu-k\sin\phi)^2}<0,
\qquad \mu\in U
\]
Hence, $G$ is strictly decreasing on $U$.

\textit{Boundary limits:}
For $\mu\in U$, the denominator lower bound gives $G(\mu)\leq \frac{\Delta}{\mu-kM}.$ Consequently, $\lim_{\mu\to+\infty}G(\mu)=0.$ Furthermore, as $\mu \to kM^+$, the integrand exhibits a non-integrable singularity near the points where $\sin\phi = M$, causing the integral to diverge. Thus, $\lim_{\mu\to kM^+} G(\mu) \to +\infty$.

\textit{Existence:}
Since $c>0$, the limit $\lim_{\mu\to kM^+} G(\mu) \to +\infty$ guarantees we can choose a value $p \in U$ sufficiently close to $kM$ such that $G(p)>c$. Similarly, $\lim_{\mu\to+\infty} G(\mu)=0$ ensures we can choose a sufficiently large $q \in U$, with $q>p$, such that $G(q)<c$. By continuity of $G$ on $[p,q]$, the intermediate value theorem guarantees the existence of
$\mu^\ast\in(p,q)$ satisfying $G(\mu^\ast)=c.$ Thus, Eq.~\eqref{mu_dwell} admits a solution
$\mu^\ast\in U$.

\textit{Uniqueness:} Suppose that two distinct solutions
$\mu_1,\mu_2\in U$ exist, with $\mu_1<\mu_2$. Since $G$ is strictly decreasing,
$G(\mu_1)>G(\mu_2),$ contradicting $G(\mu_1)=G(\mu_2)=c$. Hence, the
solution in $U$ is unique. Since no solution exists in $(-\infty,km)$ and the interval $[km,kM]$ is infeasible, $\mu^\ast$ is the unique feasible solution.
In particular, $\mu^\ast>kM$, completing the proof. Note that a similar proof can be developed for the $k<0$ case.
\end{proof}

\subsection{Constant Speed Elliptical Circumnavigation } \label{cn_const_speed}
There may also be a requirement for the quadrotor to traverse the boundary at a constant speed for persistent monitoring. For this purpose, we revisit the proposed dynamics in Eqs.~\eqref{desired_beta_dot} and \eqref{desired_phi_dot}. Therein, to ensure convergence to the elliptical boundary within a desired settling time, say $T_{cn}$, from an arbitrary initial position outside the ellipse, the control radial parameter $\lambda_{\beta}$ can be computed by setting $t_{s_{\beta}} = T_{cn}$ and $\rho = e^{-\beta_d}$ in Eq.~\eqref{lambda_beta_par}. Further, with $k=0$ and an appropriate choice of $\lambda_{\phi}$ in Eq.~\eqref{desired_phi_dot}, constant-speed circumnavigation at a desired speed $v_d$ can be achieved, as established in Proposition~\ref{prop:constant_speed_cn}.

\begin{proposition}
\label{prop:constant_speed_cn}
Consider the angular dynamics defined in Eq.~\eqref{desired_phi_dot} with $k = 0$, $\mu = 1$, and $\beta = 1$. The angular control coefficient $\lambda_\phi$ that ensures circumnavigation along the elliptical boundary with a constant speed $v_d$ is given by
\begin{equation}
\lambda_\phi = \frac{v_d}{\sqrt{(a\sin{\phi})^2 + (b\cos{\phi})^2}} 
\label{eq:lambda_phi_constant_speed}
\end{equation}
\end{proposition}

\begin{proof}
For $\beta = 1$, Eq.~\eqref{cov} reduces to $x = a \cos \phi,~y = b \sin \phi.$ Differentiating these relations with respect to time yields $
\dot{x} = -a \sin \phi \, \dot{\phi},~
\dot{y} = b \cos \phi \, \dot{\phi}.$  Accordingly, the speed of the quadrotor along the elliptical boundary can be expressed as
\begin{equation}
v = \sqrt{\dot{x}^2 + \dot{y}^2}
       = \dot{\phi} \sqrt{(a\sin{\phi})^2 + (b\cos{\phi})^2}
\label{eq:elliptic_speed}
\end{equation}
Further, Eq.~\eqref{eq:elliptic_speed} can be written as
\begin{equation}
    \dot{\phi} = \frac{v}{\sqrt{(a\sin{\phi})^2 + (b\cos{\phi})^2}}
    \label{eq:phi_intermediate}
\end{equation}
From \eqref{desired_phi_dot}, with $k = 0$, it follows that
\begin{equation}
    \dot{\phi} = \lambda_\phi \mu
    \label{eq:phi_dot_with_k_0}
\end{equation}
Enforcing the constant-speed condition $(v = v_d)$ in Eq.~\eqref{eq:phi_intermediate} and equating the resulting expression with Eq.~\eqref{eq:phi_dot_with_k_0} for  $(\mu = 1)$, corresponding to counter-clockwise motion, yields Eq.~\eqref{eq:lambda_phi_constant_speed}, thereby completing the proof.
\end{proof}

\begin{figure}
    \centering
    \includegraphics[width=0.42\linewidth]{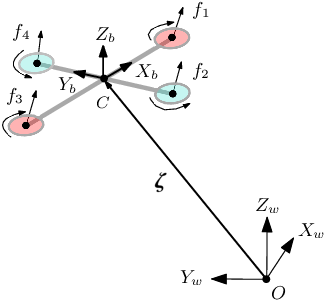}
    \caption{Quadrotor schematic with the world frame, the body-fixed frame, and rotor thrust forces.}
    \label{fig:quadrotor_frames}
\end{figure}

\section{Quadrotor Dynamics and Control}\label{sec: quadrotor model and control}
This section describes the quadrotor dynamic model and the cascaded sliding mode control framework employed to track the reference trajectories generated by the proposed guidance method.

\subsection{Quadrotor Equations of Motion} \label{6dof_model}
Consider a quadrotor of mass $m$ in the ``+'' configuration, illustrated in Fig.~\ref{fig:quadrotor_frames}. The world frame $(X_w,Y_w,Z_w)$ is defined according to the North-East-Up convention, whereas $(X_b,Y_b,Z_b)$ represents the body-fixed coordinate frame. The position of the quadrotor's center of mass is denoted by $C(x,y,z)$, and $g$ represents the gravitational acceleration. The vehicle's orientation is described by the Euler angles $(\Phi,\Theta,\psi)$, corresponding to roll, pitch, and yaw, respectively. The moments of inertia about the body-fixed axes are given by $(I_{xx},I_{yy},I_{zz})$. Using the Newton–Euler formulation, the quadrotor equations of motion are expressed as \cite{xing2023active}
\begin{align}
\ddot{x}
&= \frac{u_1}{m}\left( \cos\Phi \sin\Theta \cos\psi + \sin\Phi \sin\psi \right)
+ \frac{F_{D_x}}{m} + d_x \label{x_ddot} \\[2pt]
\ddot{y}
&= \frac{u_1}{m}\left( \cos\Phi \sin\Theta \sin\psi
- \sin\Phi \cos\psi \right)
+ \frac{F_{D_y}}{m} + d_y \label{y_ddot} \\
\ddot{z}
&= -g
+ \frac{u_1}{m}\left( \cos\Phi \cos\Theta \right)
+ \frac{F_{D_z}}{m} + d_z \label{z_ddot}  \\[2pt]
\ddot{\Phi}
&= \dot{\Theta}\dot{\psi}
\left( \frac{I_{yy}-I_{zz}}{I_{xx}} \right)
+ \frac{u_2}{I_{xx}}
+ d_{\Phi} \label{phi_ddot}  \\
\ddot{\Theta}
&= \dot{\Phi}\dot{\psi}
\left( \frac{I_{zz}-I_{xx}}{I_{yy}} \right)
+ \frac{u_3}{I_{yy}}
+ d_{\Theta} \label{theta_ddot} \\
\ddot{\psi}
&= \dot{\Phi}\dot{\Theta}
\left( \frac{I_{xx}-I_{yy}}{I_{zz}} \right)
+ \frac{u_4}{I_{zz}}
+ d_{\psi} \label{psi_ddot} 
\end{align}
The total thrust generated by the four rotors along the body-fixed $Z_b$-axis is denoted by $u_1 = \sum_{i=1}^{4} f_i$. The control inputs $u_2 = \tau_{\Phi}$, $u_3 = \tau_{\Theta}$, and $u_4 = \tau_{\psi}$ correspond to the roll, pitch, and yaw torques, respectively. The aerodynamic drag-force $\mathbf{F}_{D}$ is characterized through its components represented by $F_{D_x}$, $F_{D_y}$, and $F_{D_z}$. Furthermore, \(d_x\), \(d_y\), \(d_z\), \(d_{\Phi}\), \(d_{\Theta}\), and \(d_{\psi}\) represent bounded lumped uncertainty terms that account for unmodeled dynamics and other uncertainties, including effects associated with rotor inertia and gyroscopic coupling.

\subsubsection{Aerodynamic Drag Formulation}
The aerodynamic drag force appearing in the translational dynamics (Eqs.~\eqref{x_ddot}--\eqref{z_ddot}) is modeled based on the vehicle's velocity relative to the surrounding air. Let $\mathbf{v} = [\dot{x},\,\dot{y},\,\dot{z}]^\top$ represent the quadrotor's ground-relative velocity, and $\mathbf{v}_w = [v_{wx},\,v_{wy},\,v_{wz}]^\top$ denote the wind velocity relative to the ground, with both vectors expressed in the world frame. Accordingly, the relative air velocity is defined as
\begin{equation}
\mathbf{v}_{\mathrm{rel}} = \mathbf{v} - \mathbf{v}_w
\label{rel_wind_velocity}
\end{equation}
As the drag coefficients and effective reference areas are defined along the body axes, the relative air velocity is transformed into the body frame using $\mathbf{v}_{\mathrm{rel}}^b = R^\top\mathbf{v}_{\mathrm{rel}}$. The rotation matrix $R$, which transforms vectors from the body frame to the world frame, is defined according to the $ZYX$ Euler-angle sequence as
\begin{align}
R =
\begin{bmatrix}
c\Theta c\psi &
s\Phi s\Theta c\psi - c\Phi s\psi &
c\Phi s\Theta c\psi + s\Phi s\psi \\
c\Theta s\psi &
s\Phi s\Theta s\psi + c\Phi c\psi &
c\Phi s\Theta s\psi - s\Phi c\psi \\
-s\Theta &
s\Phi c\Theta &
c\Phi c\Theta
\end{bmatrix}
\end{align}
where $c(\cdot)$ and $s(\cdot)$ denote the cosine and sine functions, respectively. Following Ref.~\cite{raza2016autonomous}, the quadratic aerodynamic drag force expressed in the body frame is given by
\begin{equation}
\mathbf{F}_{D}^b
= -\frac{1}{2}\rho_a
\left\|\mathbf{v}_{\mathrm{rel}}^b\right\|
\begin{bmatrix}
C_{D_x}A_xv_{\mathrm{rel},x}^b, ~
C_{D_y}A_yv_{\mathrm{rel},y}^b, ~
C_{D_z}A_zv_{\mathrm{rel},z}^b
\end{bmatrix}^\top
\end{equation}
Here, $\rho_a$ denotes the ambient air density, $C_{D_x}$, $C_{D_y}$, and $C_{D_z}$ represent the drag coefficients along the body axes, and $A_x$, $A_y$, and $A_z$ denote the corresponding effective reference areas. Finally, the drag force is transformed back into the world frame as $\mathbf{F}_{D} = R\mathbf{F}_{D}^b$.

\subsubsection{Dryden Wind Turbulence Model}
The wind velocity appearing in Eq.~\eqref{rel_wind_velocity} is modelled in the world frame using the Dryden turbulence formulation as
\begin{equation}
    \mathbf{v}_w = \mathbf{v}_{\mathrm{mean}} + R \boldsymbol{\eta}_b(t) = [v_{\mathrm{mean},x}, v_{\mathrm{mean},y}, v_{\mathrm{mean},z}]^\top + R [ \eta_u(t), \eta_v(t), \eta_w(t)]^\top
    \label{wind_velocity}
\end{equation}
Here, $\mathbf{v}_{\mathrm{mean}}$ represents the mean wind velocity in the world frame, while $\boldsymbol{\eta}_b(t) = [\eta_u(t),~\eta_v(t),~\eta_w(t)]^\top$ contains the turbulent gust velocities defined along the longitudinal, lateral, and vertical axes of the vehicle body frame, respectively. The rotation matrix $R$ maps these body-frame turbulence components into the world frame. The gust components are obtained by filtering mutually independent, zero-mean Gaussian white-noise processes through linear shaping filters, as described by $ \eta_u(s) = G_u(s)n_u(s),~\eta_v(s) = G_v(s)n_v(s),~  \eta_w(s) = G_w(s)n_w(s)$. Here $n_u(t)$, $n_v(t)$, and $n_w(t)$ denote independent Gaussian white-noise inputs with zero mean, whereas $G_u(s)$, $G_v(s)$, and $G_w(s)$ are the corresponding transfer functions designed to reproduce the specified Dryden turbulence power spectral densities. 

For low-altitude operation, defined by $h \leq 1000~\mathrm{ft}$, the turbulence scale lengths and standard deviations are determined in accordance with the MIL-F-8785C specification. Since the altitude $(z)$ is expressed in metres, it is converted to feet using $h_{\mathrm{ft}} = \frac{z}{0.3048}.$ Based on this altitude, the turbulence scale lengths in the longitudinal, lateral, and vertical directions are calculated as \cite{stroe2016analysis}
\begin{equation}
L_w = h_{\mathrm{ft}}, \quad L_u = L_v =
\frac{h_{\mathrm{ft}}}
{\left(0.177+0.000823h_{\mathrm{ft}}\right)^{1.2}}
\label{scale_lengths}
\end{equation}
The calculated scale lengths are subsequently converted from feet to metres before being incorporated into the simulation. The associated turbulence intensity parameters are specified by the standard deviations
\begin{equation}
\sigma_w = 0.1W_{20}, \quad \sigma_u = \sigma_v =
\frac{\sigma_w}
{\left(0.177+0.000823h_{\mathrm{ft}}\right)^{0.4}}
\label{turbulance_std_dev}
\end{equation}
where $W_{20}$ represents the reference wind speed measured at an altitude of $20~\mathrm{ft}$. With the scale lengths (Eq.~\eqref{scale_lengths}) and turbulence standard deviations (Eq.~\eqref{turbulance_std_dev}) determined, the transfer functions of the Dryden shaping filters for the three gust components are given by
\begin{equation}
G_u(s) =
\sigma_u
\sqrt{\frac{2L_u}{\pi U_0}}
\frac{1}
{1+\left(\frac{L_u}{U_0}s\right)}, \quad G_v(s) =
\sigma_v
\sqrt{\frac{L_v}{\pi U_0}}
\frac{1+\sqrt{3}\left(\frac{L_v}{U_0}s\right)}
{\left[1+\left(\frac{L_v}{U_0}s\right)\right]^2}, \quad G_w(s) =
\sigma_w
\sqrt{\frac{L_w}{\pi U_0}}
\frac{1+\sqrt{3}\left(\frac{L_w}{U_0}s\right)}
{\left[1+\left(\frac{L_w}{U_0}s\right)\right]^2}
\end{equation}
where $U_0$ denotes the reference forward airspeed of the quadrotor.

\subsection{Cascaded Sliding-Mode Control Framework for Trajectory Tracking}

\begin{figure}
    \centering
    \includegraphics[scale = 0.74]{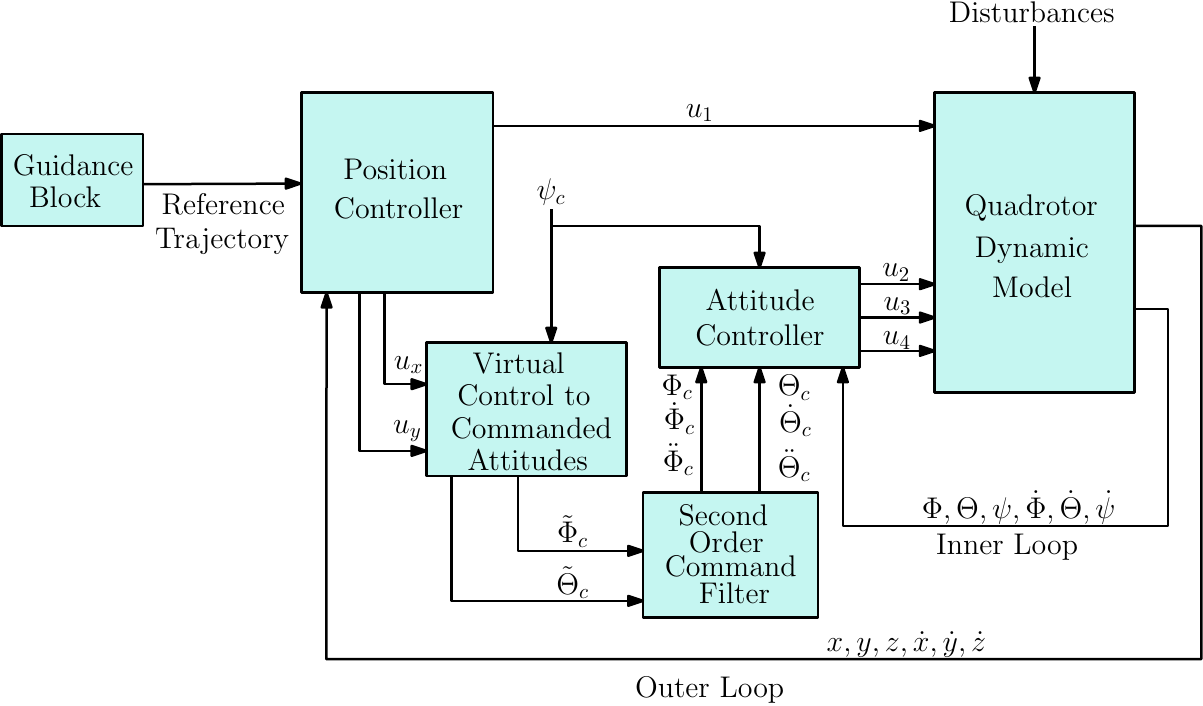}
    \caption{Block diagram of cascaded sliding mode control framework for quadrotor trajectory tracking.}
    \label{smc_block_diag}
\end{figure}

This subsection presents a cascaded sliding-mode control strategy for quadrotor trajectory tracking and stabilization, as depicted in Fig.~\ref{smc_block_diag}. The guidance block constructs the reference trajectory at a fixed altitude according to the requirements of the inspection task. The reference horizontal velocity components are obtained by differentiating Eq.~\eqref{cov} and using Eqs.~\eqref{desired_beta_dot} and \eqref{desired_phi_dot} therein, resulting in
\begin{align}
 \dot{x}_{\textrm{ref}}& = a \left( \lambda_{\beta}(e^{-\beta} - \rho)\cos{\phi}
 - \beta\lambda_{\phi}(\mu - k\sin{\phi})\sin{\phi} \right)
 \label{x_des_dot}\\
 \dot{y}_{\textrm{ref}}& = b \left( \lambda_{\beta}(e^{-\beta} - \rho)\sin{\phi}
 + \beta\lambda_{\phi}(\mu - k\sin{\phi})\cos{\phi} \right)
 \label{y_des_dot}
\end{align}
The reference position components $(x_{\textrm{ref}}, y_{\textrm{ref}})$ are obtained through integration of the corresponding velocity equations, Eqs.~\eqref{x_des_dot} and \eqref{y_des_dot}. Similarly, the desired acceleration components are calculated by taking their  time derivatives.

Adopting the approach presented in Ref.~\cite{liu2024smc}, the control structure incorporates an outer-loop position controller and an inner-loop attitude controller as shown in Fig.~\ref{smc_block_diag}. Aerodynamic drag and lumped model uncertainties are treated as bounded disturbances in the simulation. The effects of these disturbances are mitigated by the sliding-mode control law.  Let the translational position and attitude vectors be defined as $\boldsymbol{\zeta} = [x,y,z]^\top$ and $\boldsymbol{\eta} = [\Phi,\Theta,\psi]^\top$, respectively. The associated tracking errors are expressed as
$\mathbf{e}_p = \boldsymbol{\zeta}_{\mathrm{ref}} - \boldsymbol{\zeta}$ and
$\mathbf{e}_\eta = \boldsymbol{\eta}_{\mathrm{c}} - \boldsymbol{\eta}$.
To prevent Euler-angle singularities, both the actual and commanded roll and pitch angles are constrained to the interval $\left(-\pi/2,\,\pi/2\right)$. The sliding surfaces for the position and attitude variables are defined as
\begin{equation}
\mathbf{s}_p = \dot{\mathbf{e}}_p + \boldsymbol{\Lambda}_p \mathbf{e}_p, \qquad 
\mathbf{s}_\eta = \dot{\mathbf{e}}_\eta + \boldsymbol{\Lambda}_\eta \mathbf{e}_\eta
\label{sliding_surfaces}
\end{equation}
where $\boldsymbol{\Lambda}_p = \mathrm{diag}(\lambda_x, \lambda_y, \lambda_z)$ and $\boldsymbol{\Lambda}_\eta = \mathrm{diag}(\lambda_\Phi, \lambda_\Theta, \lambda_\psi)$ are positive-definite diagonal gain matrices associated with the position and attitude sliding surfaces, respectively. The evolution of each scalar sliding variable $s_i$ is governed by the continuous reaching law \cite{liu2024smc}
\begin{equation}
\dot{s}_i = -k_{1i}s_i - k_{2i}\tanh\left(\frac{s_i}{\epsilon}\right), \quad \textrm{where} \quad i \in \{x,y,z,\Phi,\Theta,\psi\}
\label{reaching_law}
\end{equation}
Here, $k_{1i}$ and $k_{2i}$ are positive reaching-law gains, while $\epsilon$ is a positive smoothing parameter that regulates the transition of the hyperbolic tangent function near the sliding surface. Position controller in the outer loop provides the total thrust command, which is derived by differentiating the vertical sliding variable $s_z$, and using the vertical translational dynamics (Eq.~\eqref{z_ddot}) together with the reaching law in Eq.~\eqref{reaching_law} as
\begin{equation}
u_1 = \frac{m}{\cos\Theta\cos\Phi} \left[ \ddot{z}_{\mathrm{ref}} + g + \lambda_z\dot{e}_z + k_{1z}s_z + k_{2z}\tanh\left(\frac{s_z}{\epsilon}\right) \right]
\label{thrust_control}
\end{equation}
To track the reference trajectory in the horizontal plane, virtual control inputs are defined based on the translational dynamics in Eqs.~\eqref{x_ddot}--\eqref{y_ddot} as follows
\begin{align}
u_x &= \frac{m}{u_1}\left[\ddot{x}_{\mathrm{ref}} + \lambda_x\dot{e}_x + k_{1x}s_x + k_{2x}\tanh\left(\frac{s_x}{\epsilon}\right)\right] \label{ux_control} \\
u_y &= \frac{m}{u_1}\left[\ddot{y}_{\mathrm{ref}} + \lambda_y\dot{e}_y + k_{1y}s_y + k_{2y}\tanh\left(\frac{s_y}{\epsilon}\right)\right]
\label{uy_control}
\end{align}
With the commanded yaw angle $\psi_{\mathrm{c}}$ specified, these virtual control signals are transformed into intermediate roll and pitch commands as
\begin{equation}
\tilde{\Phi}_{\mathrm{c}} = \sin^{-1}\left( u_x\sin\psi_{\mathrm{c}} - u_y\cos\psi_{\mathrm{c}} \right), \qquad
\tilde{\Theta}_{\mathrm{c}} = \sin^{-1}\left( \frac{u_x\cos\psi_{\mathrm{c}} + u_y\sin\psi_{\mathrm{c}}}{\cos\tilde{\Phi}_{\mathrm{c}}} \right)
\label{angles_command}
\end{equation}

To obtain the derivatives of the commanded roll and pitch angles required by the attitude controller, direct differentiation of \(\tilde{\Phi}_{\mathrm{c}}\) and \(\tilde{\Theta}_{\mathrm{c}}\) in Eq.~\eqref{angles_command} would require differentiating the virtual control expressions in Eqs.~\eqref{ux_control} and \eqref{uy_control}, resulting in cumbersome analytical expressions. To avoid this complexity, the intermediate roll and pitch commands are passed through second-order command filters, which generate the corresponding attitude references and their time derivatives \cite{liu2023command} as follows
\begin{equation}
\ddot{\alpha}_{\mathrm{c}} + 2\xi\omega_n\dot{\alpha}_{\mathrm{c}} + \omega_n^2\alpha_{\mathrm{c}} = \omega_n^2\tilde{\alpha}_{\mathrm{c}}, \qquad \alpha \in \{\Phi,\Theta\}
\label{command_filter}
\end{equation}
In Eq.~\eqref{command_filter}, $\xi$ and $\omega_n$ represent the filter damping ratio and natural frequency, respectively. The resulting filtered signals, $\alpha_{\mathrm{c}}$, $\dot{\alpha}_{\mathrm{c}}$, and $\ddot{\alpha}_{\mathrm{c}}$, provide smooth attitude commands together with their first and second-order time derivatives. The control torques are obtained by applying the reaching law in Eq.~\eqref{reaching_law} to the attitude sliding surfaces and using the rotational dynamics in Eqs.~\eqref{phi_ddot}--\eqref{psi_ddot}, yielding
\begin{align}
u_2 &= I_{xx}\left[\ddot{\Phi}_{\mathrm{c}} + \lambda_{\Phi}\dot{e}_{\Phi} + k_{1\Phi}s_{\Phi} + k_{2\Phi}\tanh\left(\frac{s_{\Phi}}{\epsilon}\right)\right] - (I_{yy}-I_{zz})\dot{\Theta}\dot{\psi} \label{roll_control} \\
u_3 &= I_{yy}\left[\ddot{\Theta}_{\mathrm{c}} + \lambda_{\Theta}\dot{e}_{\Theta} + k_{1\Theta}s_{\Theta} + k_{2\Theta}\tanh\left(\frac{s_{\Theta}}{\epsilon}\right)\right] - (I_{zz}-I_{xx})\dot{\Phi}\dot{\psi} \label{pitch_control} \\
u_4 &= I_{zz}\left[\ddot{\psi}_{\mathrm{c}} + \lambda_{\psi}\dot{e}_{\psi} + k_{1\psi}s_{\psi} + k_{2\psi}\tanh\left(\frac{s_{\psi}}{\epsilon}\right)\right] - (I_{xx}-I_{yy})\dot{\Phi}\dot{\Theta} \label{yaw_control}
\end{align}

\section{Qualitative Comparison With Prior Works}\label{sec: qualitative comparison}
This section presents a qualitative comparison between the proposed method and existing elliptical circumnavigation guidance strategies, emphasizing differences in methodology, control inputs and parameters involved, and achievable trajectory modes. A summary of this comparison is provided in Table~\ref{comp_table}. As evident from the table, the proposed method yields analytical expressions for the quadrotor velocity components that are algebraically simpler, thereby facilitating straightforward practical implementation. Notably, the guidance strategies in Refs.~\cite{lawrence2008lyapunov} and \cite{shivam2022vector} are specifically designed for fixed-wing UAVs, making them unsuitable for hovering at or switching between viewpoints. Furthermore, with regard to sectorial dwell time constraints along the boundary, control over the quadrotor speed is required. Although the methods in Refs.~\cite{shao2023distance} and \cite{cheng2010application} exhibit variations in speed during elliptical circumnavigation, no explicit control over the speed is provided. Consequently, the proposed framework is the only approach that enables the vehicle to dwell for a prescribed duration within a specified angular sector of the elliptical boundary. Moreover, unlike existing approaches that are typically limited to a single trajectory mode, the proposed framework has the flexibility to achieve four distinct trajectory modes by choosing a single bifurcation parameter $ (\mu)$ and a constant $(k)$ within a unified formulation.

\setlength{\tabcolsep}{8pt} 
\renewcommand{\arraystretch}{1.2} 

\begin{table}
\caption{Qualitative comparison with prior methods}
\label{comp_table}
\centering
\begin{tabular}{p{1.1cm} p{3.9cm} p{3.5cm}@{\hspace{1.2cm}} p{2.6cm} p{1.3cm} p{1.5cm}}
\hline
\hline
\multirow{3}{*}{Reference} & 
\multirow{3}{*}{Guidance methodology} & 
\multirow{3}{*}{Control inputs and parameters} & 
\multicolumn{3}{c}{Trajectory modes} \\
\cline{4-6}
& & & Reaching a viewpoint \& switching & \multicolumn{2}{c}{Circumnavigation} \\
\cline{5-6}
& & &  between viewpoints & Constant speed & Sectoral dwell time feature \\
\hline
\cite{lawrence2008lyapunov} & UAV velocity components are determined in analytical form by transforming a circular Lyapunov vector field through an inverse mapping function, resulting in an elliptical attractor.  & The desired course angle is deduced from the velocity components each containing four control parameters, and is tracked using a controller with a turning rate gain.  & No & Yes & No \\

\cite{shivam2022vector}  & UAV course angle vector field is expressed as a  combination of tangent to the ellipse and a non-linear arcsine function of position error  with respect to the ellipse.  & The commanded course angle incorporates a vector field smoothing gain, and is tracked using a first-order controller with a tracking gain.   &       No      &    Yes   &  No   \\

 \cite{gandolfo2016stable}  & Velocity components are obtained from a hyperbolic tangent function of the position and orientation errors of the quadrotor relative to the closest point on an ellipse, combined with a desired path-tangent velocity vector.   & Velocity components are expressed in closed-form, involving eight control parameters where four parameters limit the maximum control effort and the other four control the tracking error.  &       No      &    Yes   &  No   \\ 

 \cite{shao2023distance}   & Vehicle heading direction is generated by combining the radial correction toward the desired ellipse with the tangential direction along the ellipse, and the resulting heading is tracked using an angular-rate controller.   & Velocity and heading commands are in closed-form, incorporating a tangential speed scaling gain and a heading tracking gain. &  No       &     No      &  No   \\

 \cite{cheng2010application}   & Robot velocity components are generated by converting a circular limit cycle into an elliptical one through anisotropic scaling. 
  & Closed-form velocity inputs involve three control parameters determining orbital direction, convergence magnitude, and ellipse shape, respectively.
  &       No       &     No    & No    \\ 

  Proposed method   & Quadrotor velocity is obtained using saddle-node infinite period bifurcation in quadrotor position variables $\beta$ and $\phi$.  & Closed-form velocity inputs are given by Eqs.~\eqref{x_des_dot} and \eqref{y_des_dot} and depend on a bifurcation parameter ($\mu$), a constant ($k$), a radial control parameter ($\lambda_{\beta}$), and an angular control coefficient ($\lambda_{\phi}$) to regulate the settling time.           &      Yes            &  Yes    &    Yes \\
\hline \hline
\end{tabular}
\end{table}

\section{Simulation Studies}\label{sec:simulation results}
In this section, the proposed method is validated through numerical simulations of the four boundary-task scenarios, wherein the reference trajectories generated by the proposed method are tracked using a 6-DOF quadrotor model. The quadrotor's physical and aerodynamic parameters are taken from Ref.~\cite{raza2016autonomous}. The vehicle mass and gravitational acceleration are $m = 2.2~\mathrm{kg}$ and $g = 9.81~\mathrm{m/s^2}$, respectively, while the principal moments of inertia are $[I_{xx}, I_{yy}, I_{zz}] = [0.0150, 0.0138, 0.0205]~\mathrm{kg\cdot m^2}$. The air density is $\rho_a = 1.225~\mathrm{kg/m^3}$, and the drag coefficients are assumed to be identical along all three axes, with $C_{D_x} = C_{D_y} = C_{D_z} = 1.05$. The corresponding reference areas are $A_x = A_y = 0.0136~\mathrm{m^2}$ and $A_z = 0.0256~\mathrm{m^2}$. The commanded yaw angle is $\psi_c = 0$. To account for model uncertainties, the disturbance terms are specified as $d_x = d_y = d_z = d_{\Phi} = d_{\theta} = 0.1\sin(t)$ and $d_{\psi} = 0.01\sin(t)$. Wind disturbances are modeled using the Dryden wind model, with a reference wind speed of $W_{20} = 7.7167~\mathrm{m/s}$ (15 knots), a forward speed of $U_0 = 5~\mathrm{m/s}$, and mean wind components of $v_{\mathrm{mean},x} = v_{\mathrm{mean},y} = v_{\mathrm{mean},z} = 0.5~\mathrm{m/s}$.

For the sliding-mode controller, the sliding-surface gains are chosen as $\lambda_x = \lambda_y = \lambda_z = 2$, $\lambda_{\Phi} = \lambda_{\theta} = 10$, and $\lambda_{\psi} = 5$. The linear reaching gains are set to $k_{1x} = k_{1y} = k_{1z} = 3$, $k_{1\Phi} = k_{1\theta} = 15$, and $k_{1\psi} = 8$. The nonlinear reaching gains are selected as $k_{2x} = k_{2y} = 1$, $k_{2z} = 2$, $k_{2\Phi} = k_{2\theta} = 4$, and $k_{2\psi} = 3$. The smoothing parameter is set to $\epsilon = 0.4$. Finally, the command filter parameters are chosen as $\xi = 1$ and $\omega_n = 25$. All simulations are carried out at a fixed altitude of $z_{\mathrm{ref}} = 5~\mathrm{m}$. The centre of the elliptical boundary is at the origin, with semi-major and semi-minor axes $(a,b) = (6~\mathrm{m}, 4~\mathrm{m})$. The neighborhood bounds around the desired values of $\beta$ and $\phi$ are chosen as $\epsilon_\beta = 0.001$ and $\epsilon_\phi = 0.005~\mathrm{deg}$, respectively. The elliptic radial parameter $\rho = 0.3679$ is computed using $\beta_d = 1$ in Eq.~\eqref{p_parameter}.

\subsection{Case 1: Rendezvous with a Viewpoint}
In this case, with the quadrotor initial position at $(x_0, y_0) = (9~\mathrm{m}, 8~\mathrm{m})$, the objective is to reach a desired viewpoint on the boundary, specified by the line-of-sight angle $\theta_d = 160~\mathrm{deg}$, within a settling time $T_{vp} = 20~\mathrm{s}$.  The corresponding desired elliptic angle $\phi_d = 151.3674~\mathrm{deg}$ is computed using Eq.~\eqref{main_phi_d}. The associated guidance parameters $\lambda_{\beta} = 1.1083$, $\lambda_{\phi} = 0.6369$, $k = -1$, and $\mu = -0.4792$ are determined based on the analysis presented in Sec.~\ref{reach_vp}. The resulting quadrotor trajectory converges to the desired viewpoint on the elliptic boundary as shown in Fig.~\ref{case1_traj}. The variation in the elliptic radial distance and the elliptic angle with respect to the line-of-sight angle are shown in Figs.~\ref{case1}(a) and \ref{case1}(b), respectively. The variables attain their desired values of $\beta_d = 1~\mathrm{m}$ and $\theta_d = 160~\mathrm{deg}$ with negligible error, and within the prescribed settling time of $T_{vp} = 20~\mathrm{s}$. The altitude response, shown in Fig.~\ref{case1}(c), closely tracks the reference altitude $z_{\mathrm{ref}} = 5~\mathrm{m}$. The commanded and attained pitch and roll responses are presented in Figs.~\ref{case1}(d) and \ref{case1}(e), respectively. The corresponding velocity profile in Fig.~\ref{case1}(f) illustrates the variation in quadrotor speed as it approaches the desired viewpoint.

Further simulation studies are performed for four additional scenarios in which the quadrotor initial position is at $(x_0, y_0) = (-9~\mathrm{m}, 7~\mathrm{m})$, while the desired viewpoints are selected in different quadrants on the boundary. The set of desired viewpoints, along with the corresponding settling times and the associated guidance parameters obtained using the procedure described in Sec.~\ref{reach_vp}, is summarized in Table~\ref{table_for_case1}. The resulting quadrotor trajectories are shown in Fig.~\ref{case1a_traj}. The elliptic radial profiles in Fig.~\ref{case1a}(a) show that for all the scenarios, the quadrotor reaches the elliptical boundary within the specified settling time. The time responses of the line-of-sight and elliptic angles, plotted in Figs.~\ref{case1a}(b) and (c), show convergence to desired viewpoints with negligible errors. The corresponding velocity profiles are presented in Fig.~\ref{case1a}(d).

\begin{table*}[]
\caption{Simulation Parameters For Case 1: Viewpoints in Different Quadrants}
\label{table_for_case1}
\centering
\begin{tabular}{c c c c c c c c}
\hline \hline
   $\theta_d$  &  $T_{vp}$  &  $(x_d, y_d)$  &   $\phi_d$  & \multirow{2}{*}{$\lambda_\beta$}   & \multirow{2}{*}{$\lambda_\phi$}   & \multirow{2}{*}{$k$}     & \multirow{2}{*}{$\mu$} \\
 $(\textrm{deg})$  & $(\textrm{s})$  &  $(\textrm{m})$   &  $(\textrm{deg})$   &  &   &  &      \\ \hline
 
30 & 15 & (4.5356, 2.6186)  &  40.8934   & 1.4309  & 1.0594  &  1  &  0.6547                 \\ 

150 & 18 & (-4.5356, 2.6186)  &  139.1066   & 1.1924  & 0.5517  &  -1  &  -0.6547                 \\ 

195 & 12 & (-5.5672, -1.4917))  & 201.8964    & 1.7886  & 0.8491  &  -1  &0.3729                   \\ 

350 & 30 &  (5.8005, -1.0228) &  345.1851   &0.7154   &0.3652   &  1  &  -0.2557                 \\ \hline \hline
\end{tabular}
\end{table*}

\begin{figure}
  \centering
  \includegraphics[width=0.52\textwidth]{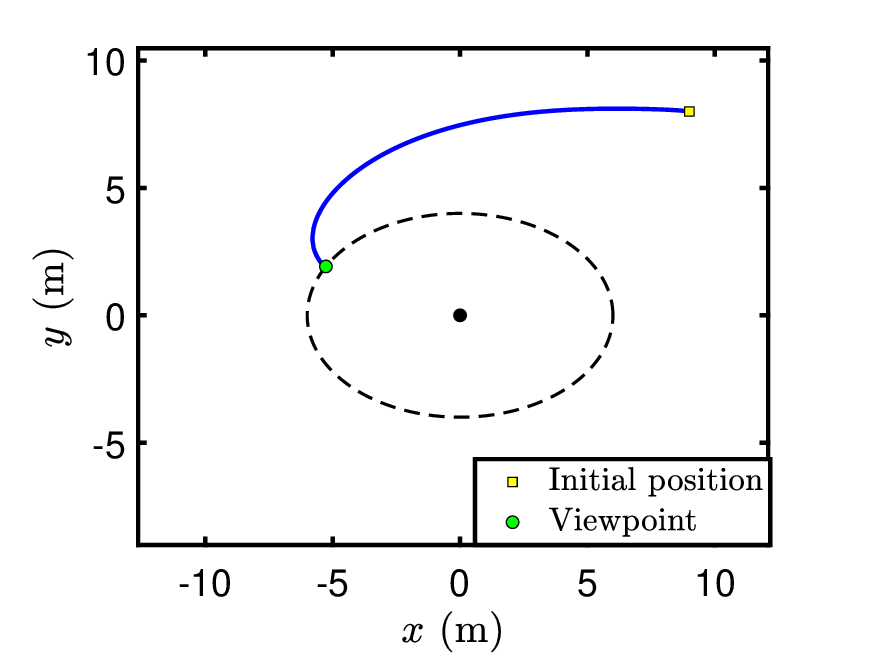}
  \caption{Quadrotor trajectory for Case 1.}
  \label{case1_traj}
\end{figure}

 \begin{figure}
	\centering
	\subfloat[{ Elliptic radius}]
	{\includegraphics[width=0.5\textwidth]{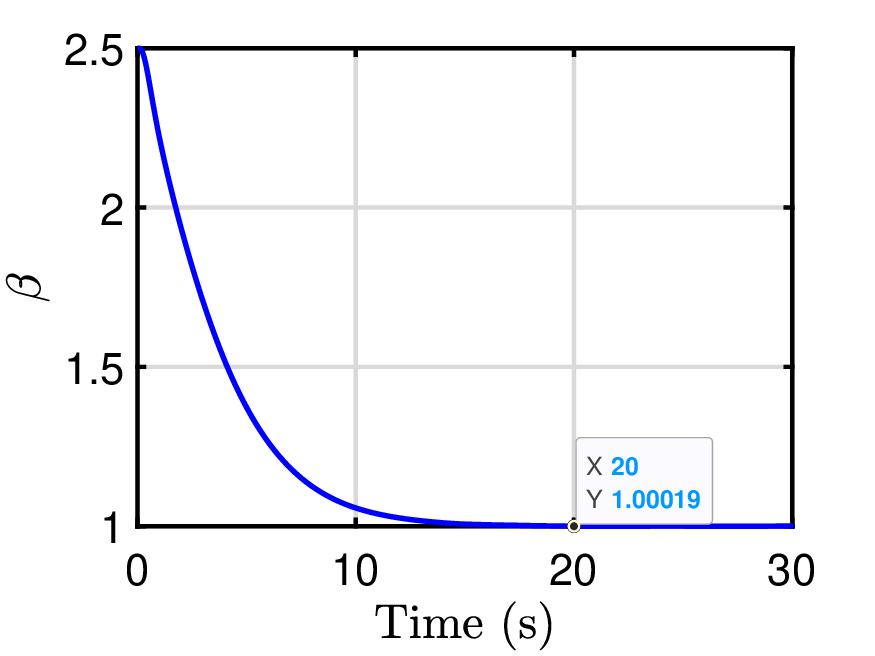}}
		\subfloat[{Line-of-sight and elliptic angles}]
	{\includegraphics[width=0.5\textwidth]{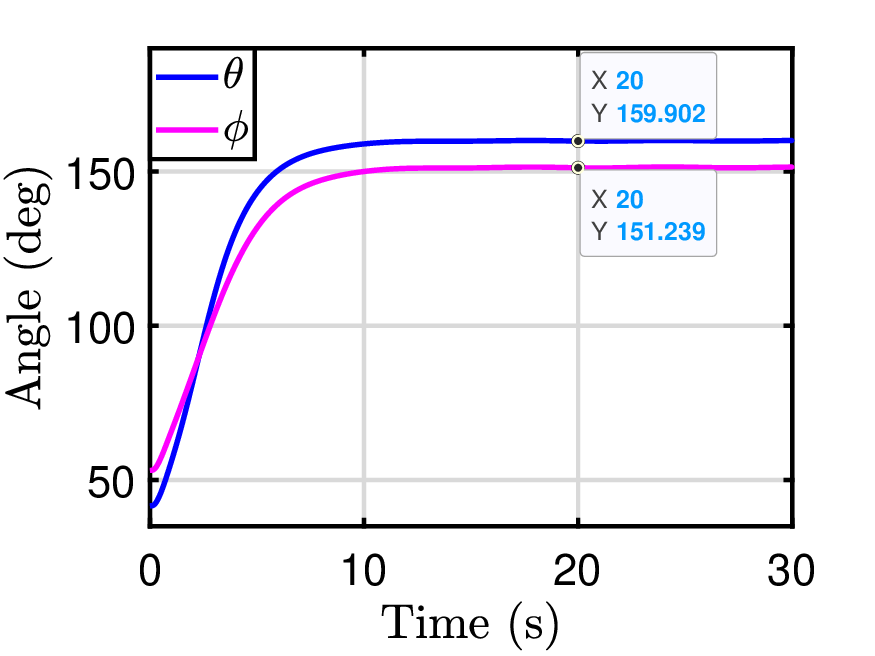}}\\
	\subfloat[{Altitude}]
	{\includegraphics[width=0.5\textwidth]{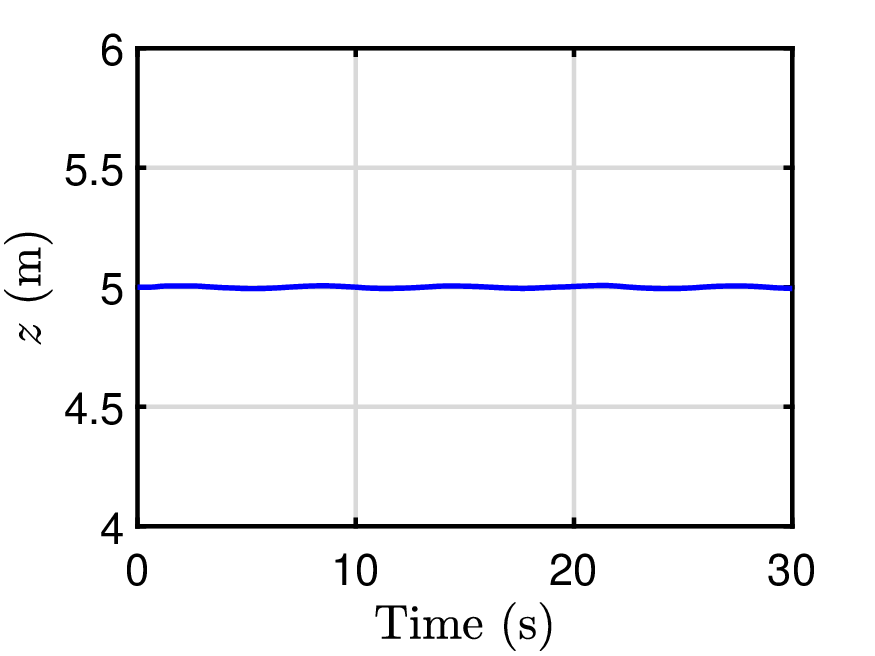}}
	\subfloat[{Roll angle} ]
	{\includegraphics[width=0.5\textwidth]{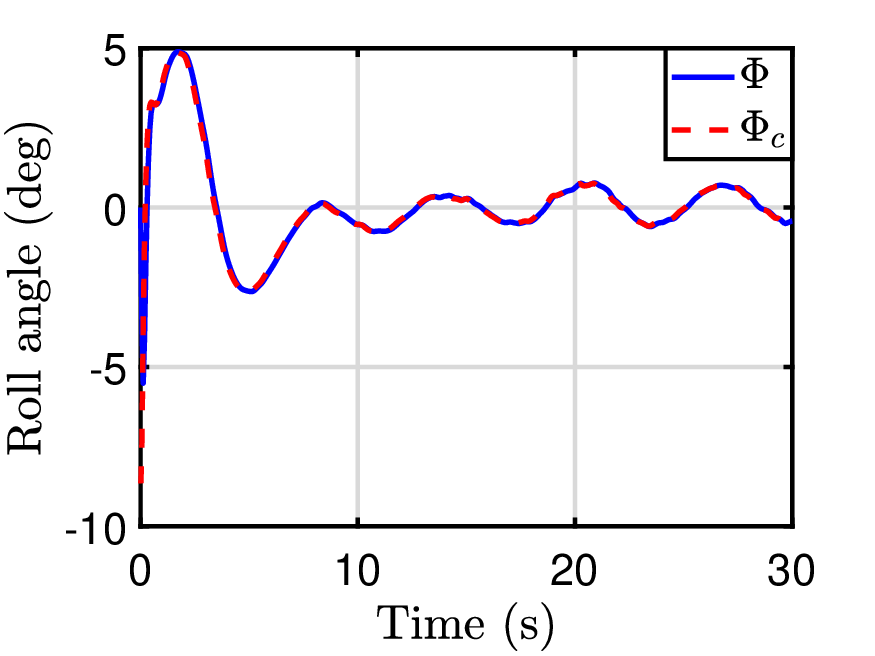}}\\
    \subfloat[{Pitch angle}]
	{\includegraphics[width=0.5\textwidth]{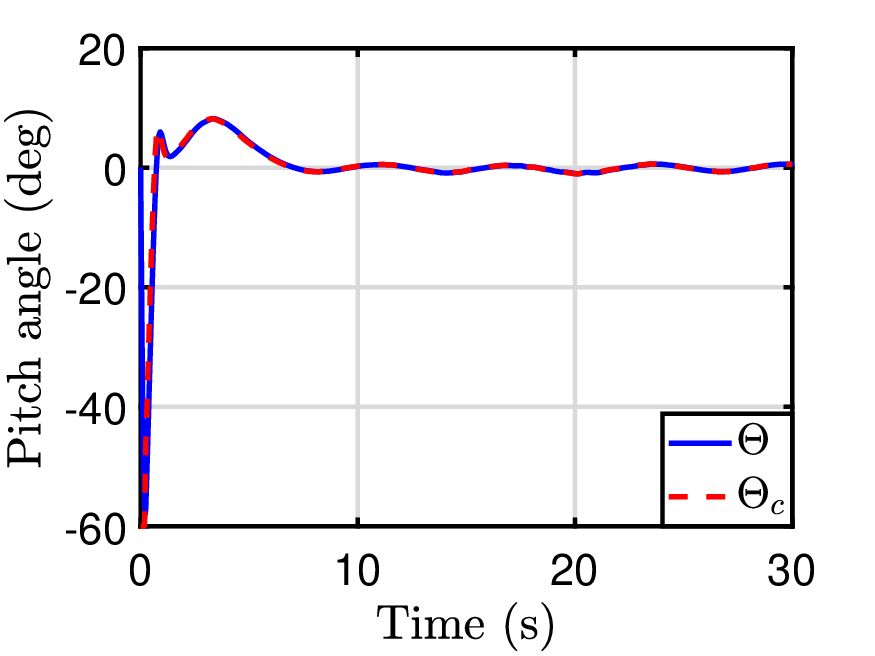}}
		\subfloat[{Velocity}]
	{\includegraphics[width=0.5\textwidth]{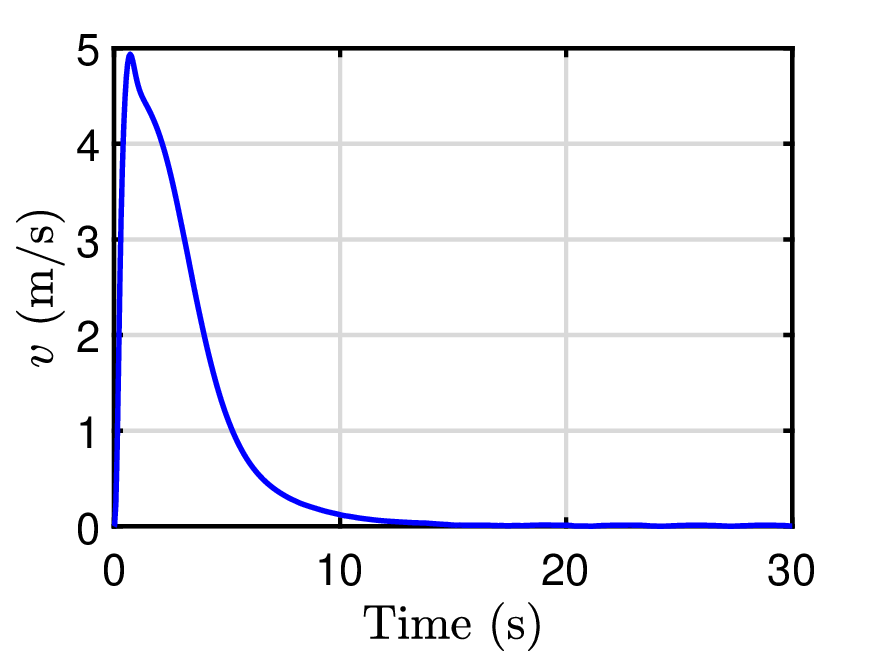}}
	\caption{Results for Case 1: Trajectory properties.}
	 \label{case1}
\end{figure}

\begin{figure}
  \centering
  \includegraphics[width=0.52\textwidth]{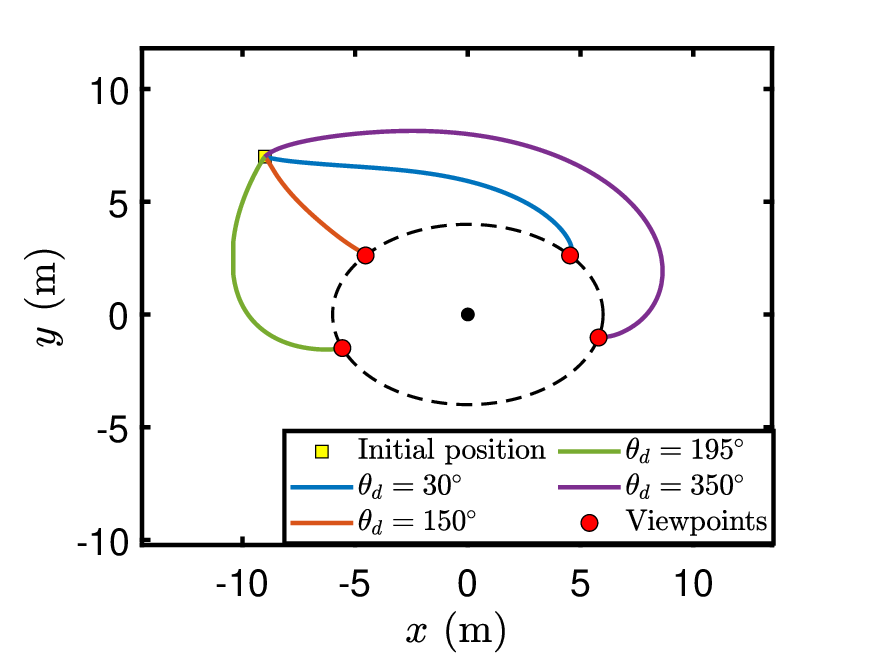}
  \caption{Results for Case 1: Viewpoints in different quadrants. }
  \label{case1a_traj}
\end{figure}

 \begin{figure}
	\centering
	\subfloat[{ Elliptic radii}]
	{\includegraphics[width=0.5\textwidth]{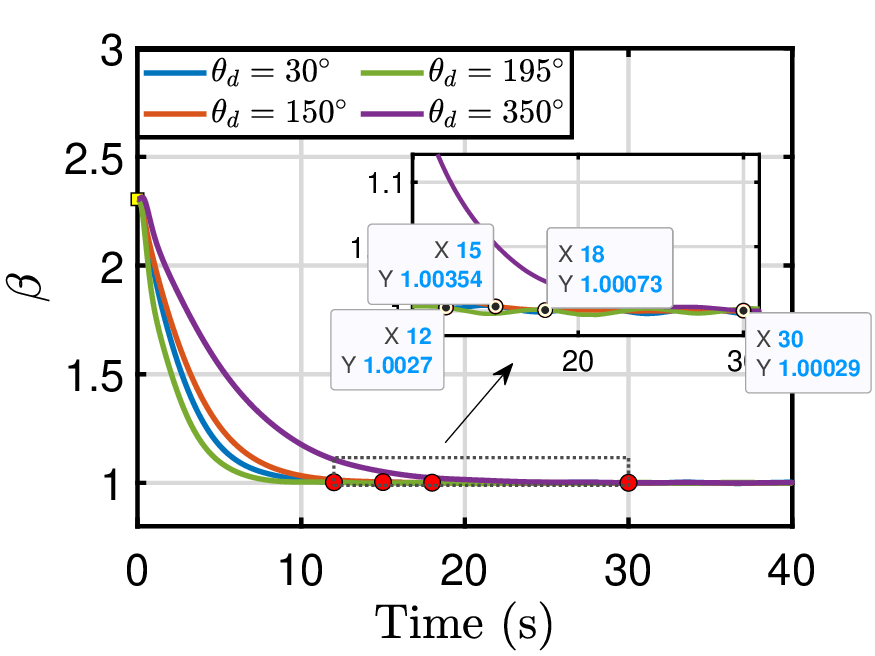}}
		\subfloat[{Line-of-sight angles}]
	{\includegraphics[width=0.5\textwidth]{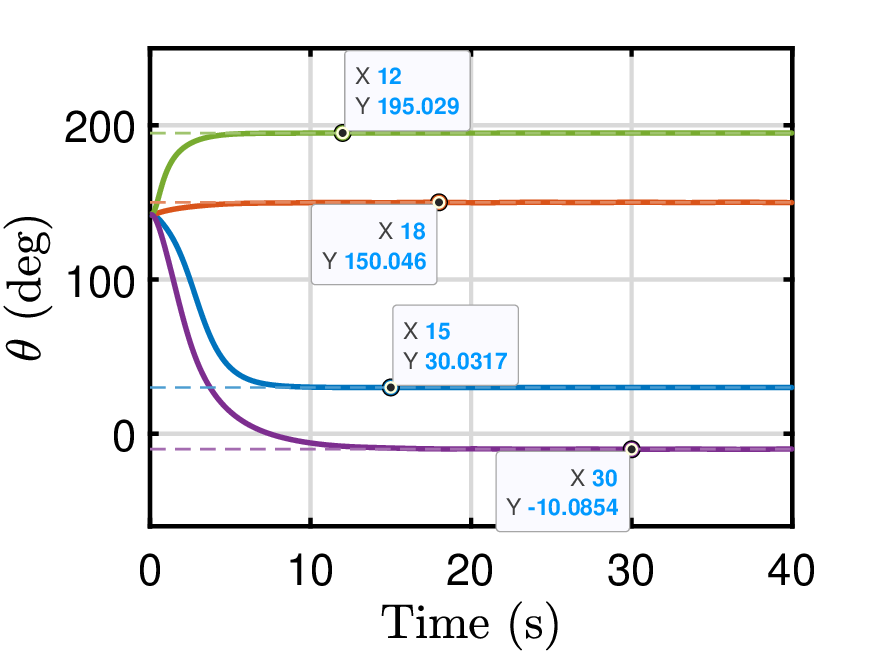}}\\
	\subfloat[{Elliptic angles }]
	{\includegraphics[width=0.5\textwidth]{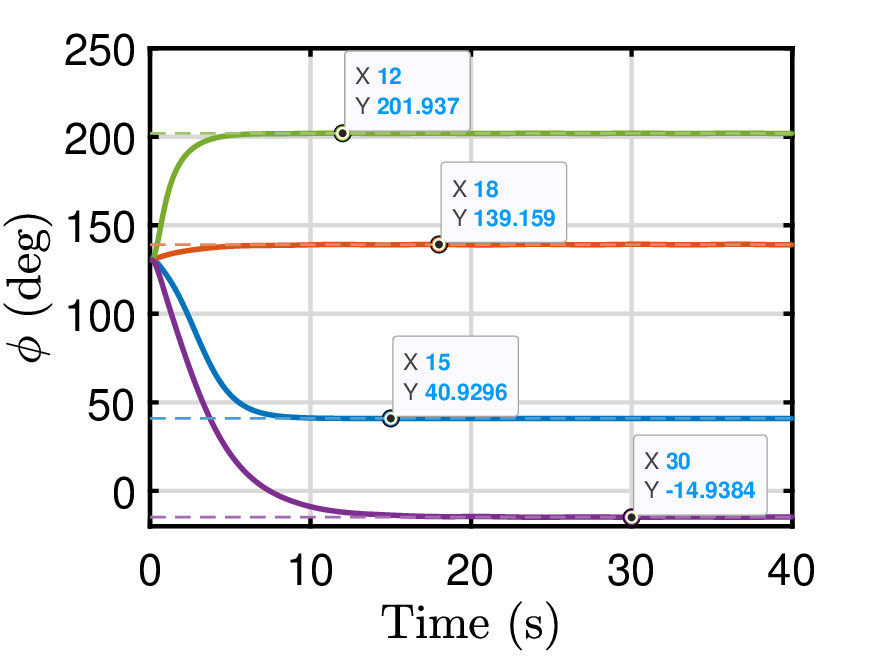}}
		\subfloat[{Velocities }]
	{\includegraphics[width=0.5\textwidth]{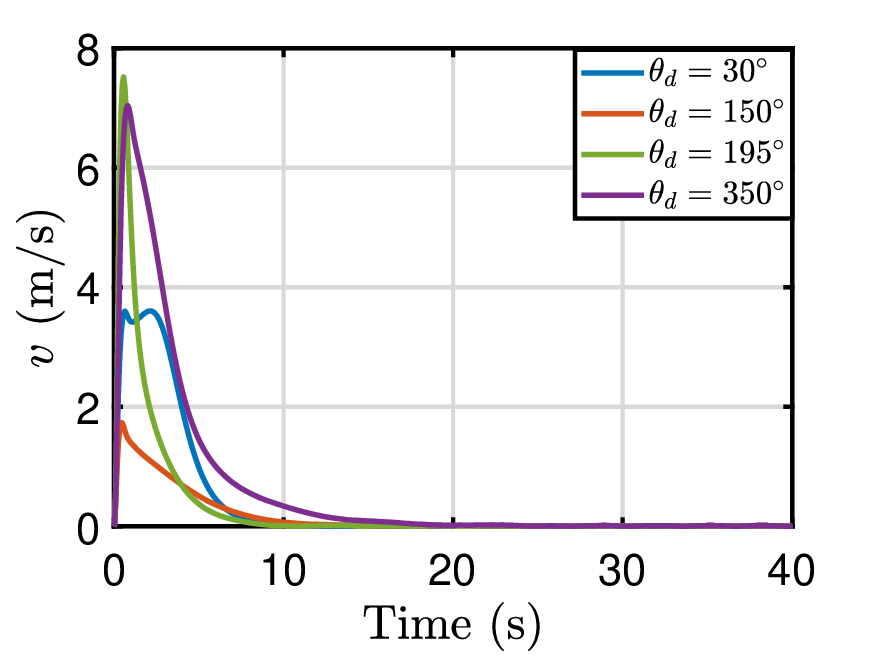}}
	\caption{Results for Case 1: Properties of trajectories leading to different quadrants.}
	 \label{case1a}
\end{figure}

\subsection{Case 2: Switching Between Viewpoints}
Here, the quadrotor initially located at $(x_0, y_0) = (8.5~\mathrm{m}, 7~\mathrm{m})$, is required to reach viewpoint A ($\theta_d = 152~\mathrm{deg}$) within a settling time of $T_{vp} = 22~\mathrm{s}$, followed by a switch to viewpoint B ($\theta_d = 210~\mathrm{deg}$) on the elliptical boundary within the duration $T_{sw} = 38~\mathrm{s}$. For reaching viewpoint A, the guidance parameters $(\lambda_{\beta} = 0.9665, \rho = 0.3679, \lambda_{\phi} = 0.6787, k = -1, \mu = -0.6235)$ are computed using the results in Sec.~\ref{reach_vp}. For switching to viewpoint B, the corresponding parameters $(\lambda_{\phi} = 0.3575, k = -1, \mu = 0.6547)$ are obtained following the method outlined in Sec.~\ref{switch_vp}. The resulting quadrotor trajectory is presented in Fig.~\ref{case2_traj}. The elliptic radial distance is shown in Fig.~\ref{case2}(a), while the variation in the line-of-sight angle is depicted in Fig.~\ref{case2}(b), indicating that the vehicle reaches viewpoint A at $t = 22~\mathrm{s}$ and subsequently arrives at viewpoint B at $t = T_{vp} + T_{sw} = 60~\mathrm{s}$, in accordance with the prescribed switching time. The corresponding elliptic angle variation is illustrated in Fig.~\ref{case2}(c), and the vehicle’s velocity profile is presented in Fig.~\ref{case2}(d).

\begin{figure}
  \centering
  \includegraphics[width=0.52\textwidth]{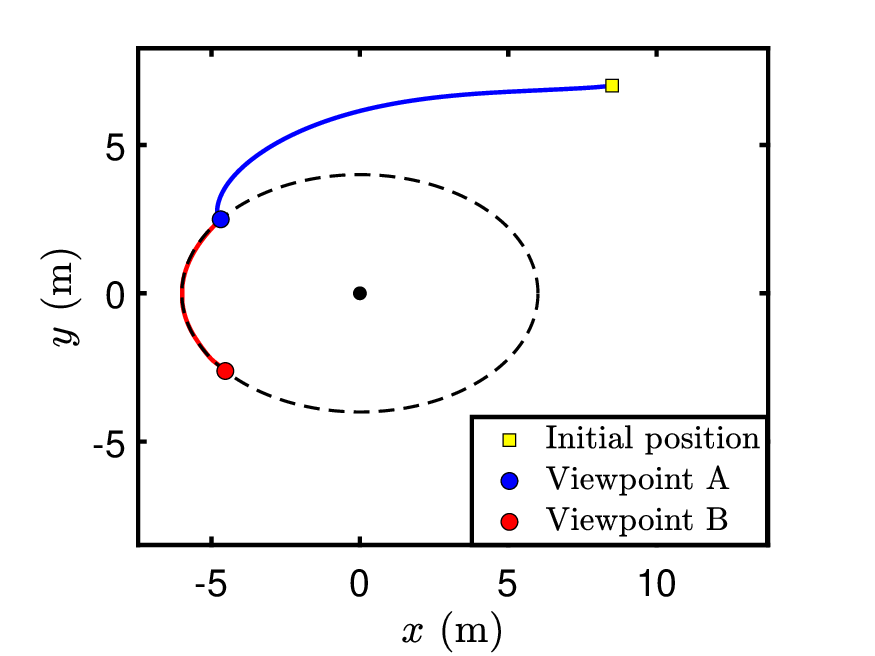}
  \caption{Quadrotor trajectory for Case 2.}
  \label{case2_traj}
\end{figure}

 \begin{figure}
	\centering
	\subfloat[{ Elliptic radius}]
	{\includegraphics[width=0.5\textwidth]{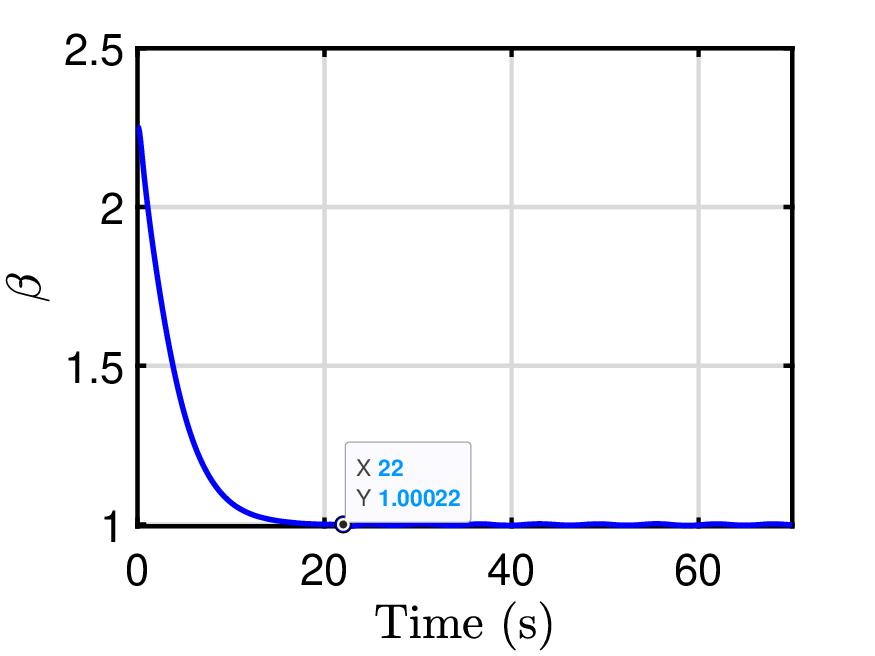}}
		\subfloat[{Line-of-sight angle}]
	{\includegraphics[width=0.5\textwidth]{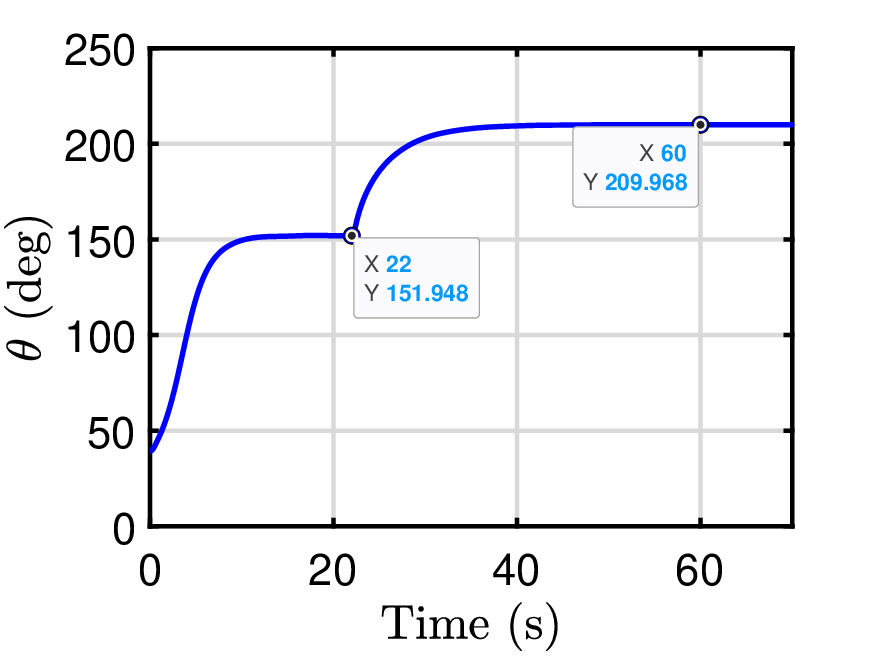}}\\
	\subfloat[{Elliptic angle}]
	{\includegraphics[width=0.5\textwidth]{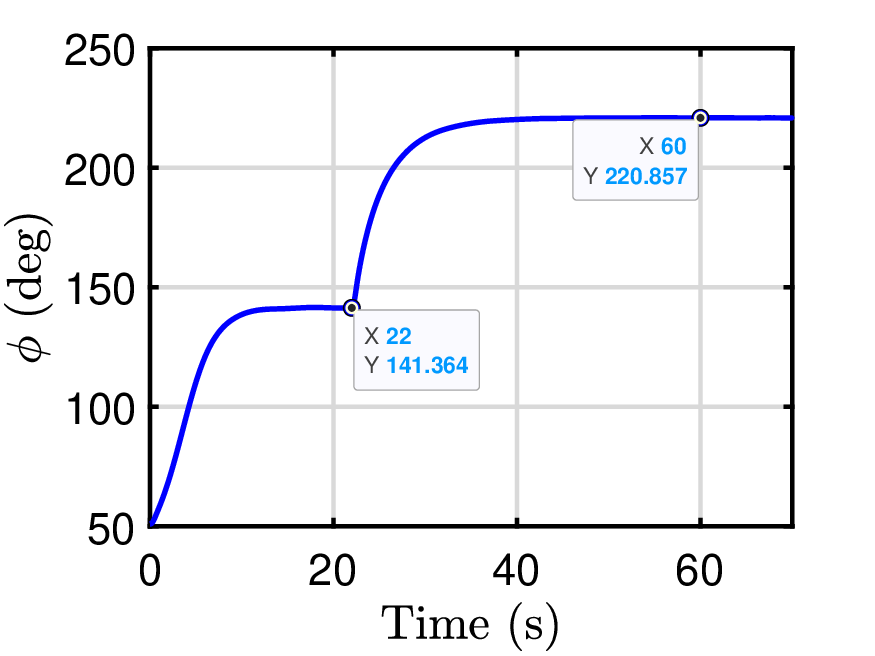}}
		\subfloat[{Velocity}]
	{\includegraphics[width=0.5\textwidth]{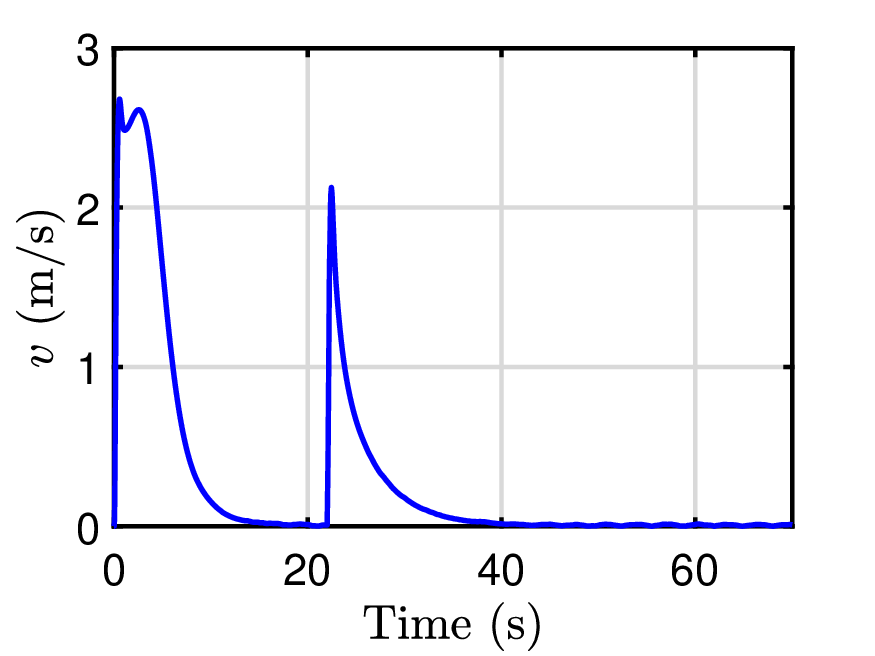}} 
	\caption{Results for Case 2: Switching between viewpoints along the elliptical boundary.}
	 \label{case2}
\end{figure}

\subsection{Case 3: Constant Speed Elliptical Circumnavigation}
This case considers a scenario in which the quadrotor is initially at $(x_0, y_0) = (-8~\mathrm{m}, -8~\mathrm{m})$ and is required to converge to the elliptical boundary within a settling time of $T_{cn} = 24.5~\mathrm{s}$. Subsequently, it is desired to fly along the elliptical boundary in the counter-clockwise direction at a constant speed of $v_d = 1~\mathrm{m/s}$. The associated guidance parameters $(\lambda_{\beta} = 0.8908, \rho = 0.3679, k = 0, \mu = 1)$ are computed using the results described in Sec.~\ref{cn_const_speed}, where $\mu = 1$ is selected to achieve counter-clockwise motion. The resulting trajectory is shown in Fig.~\ref{case3_traj}. The elliptic radial distance profile, depicted in Fig.~\ref{case3}(a), indicates that the quadrotor reaches the elliptical boundary within the stipulated time with negligible error. The elliptic angle variation is presented in Fig.~\ref{case3}(b). Figure~\ref{case3}(c) shows that the quadrotor closely tracks the reference altitude. The roll and pitch attitude responses are plotted in Figs.~\ref{case3}(d) and \ref{case3}(e), respectively, while the velocity profile in Fig.~\ref{case3}(f) demonstrates that, after settling on the elliptical boundary, the quadrotor maintains almost the desired speed.

Next, two additional simulation scenarios are considered with a common initial quadrotor position of $(x_0, y_0) = (-10~\mathrm{m}, 8~\mathrm{m})$. The scenarios differ in the specified settling times for convergence to the boundary and the desired speeds for subsequent motion along the boundary. The corresponding simulation parameters are summarized in Table~\ref{table_for_case3}. Here, $\mu = 1$ corresponds to counter-clockwise motion. The resulting quadrotor trajectories, depicted in Fig.~\ref{case3a}(a), converge to the elliptical boundary within the prescribed settling times, as shown in Fig.~\ref{case3a}(b). The corresponding elliptic angle variations are presented in Fig.~\ref{case3a}(c), while the velocity profiles shown in Fig.~\ref{case3a}(d) indicate that, upon reaching the elliptical boundary, the quadrotor maintains the specified constant speed  with negligible fluctuations.

\begin{table}[]
\caption{Simulation Parameters for Case 3}
\label{table_for_case3}
\centering
\begin{tabular}{c c c c}
\hline \hline
$T_{\textrm{cn}}$ (s) & $v_d$ (m/s) & $\lambda_{\beta}$ & $\mu$ \\ \hline
40 & 0.5 & 0.5631 & 1 \\ 
33 & 1.5 & 0.6825 & 1 \\ \hline \hline
\end{tabular}
\end{table}

\begin{figure}
  \centering
  \includegraphics[width=0.52\textwidth]{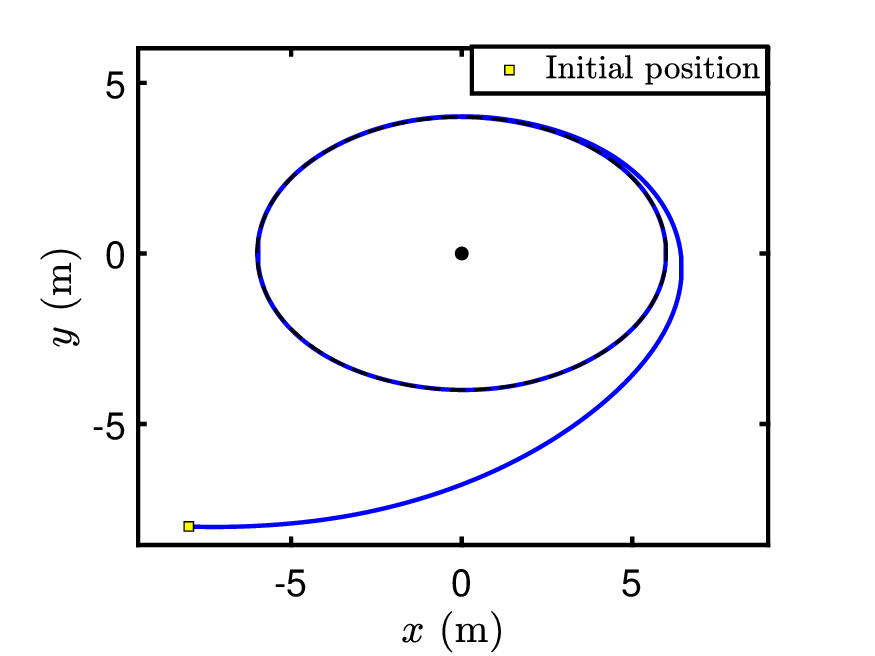}
  \caption{Quadrotor trajectory for Case 3.}
  \label{case3_traj}
\end{figure}

 \begin{figure}
	\centering
	\subfloat[{ Elliptic radius}]
	{\includegraphics[width=0.5\textwidth]{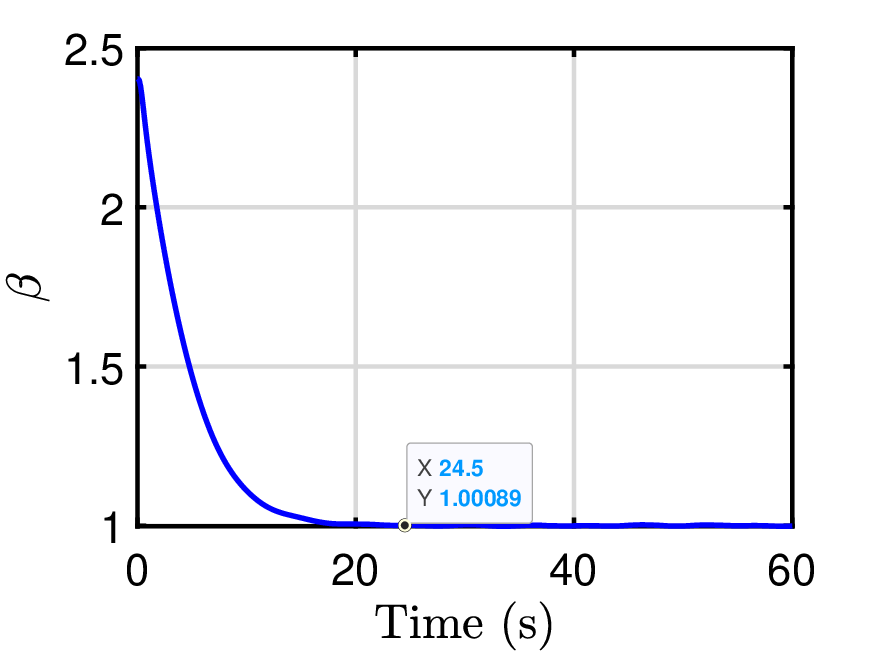}}
		\subfloat[{Elliptic angle}]
	{\includegraphics[width=0.5\textwidth]{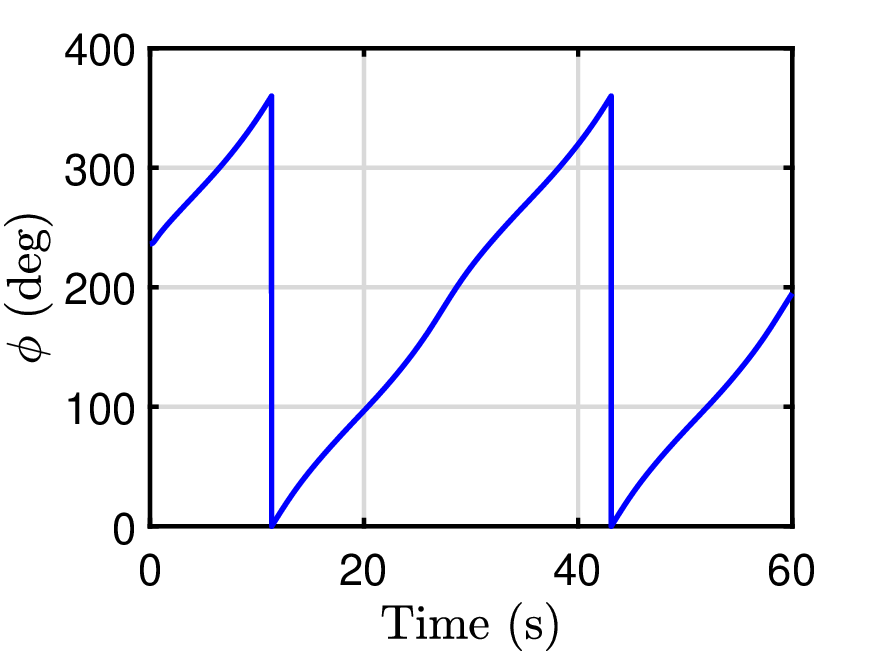}}\\
	\subfloat[{Altitude }]
	{\includegraphics[width=0.5\textwidth]{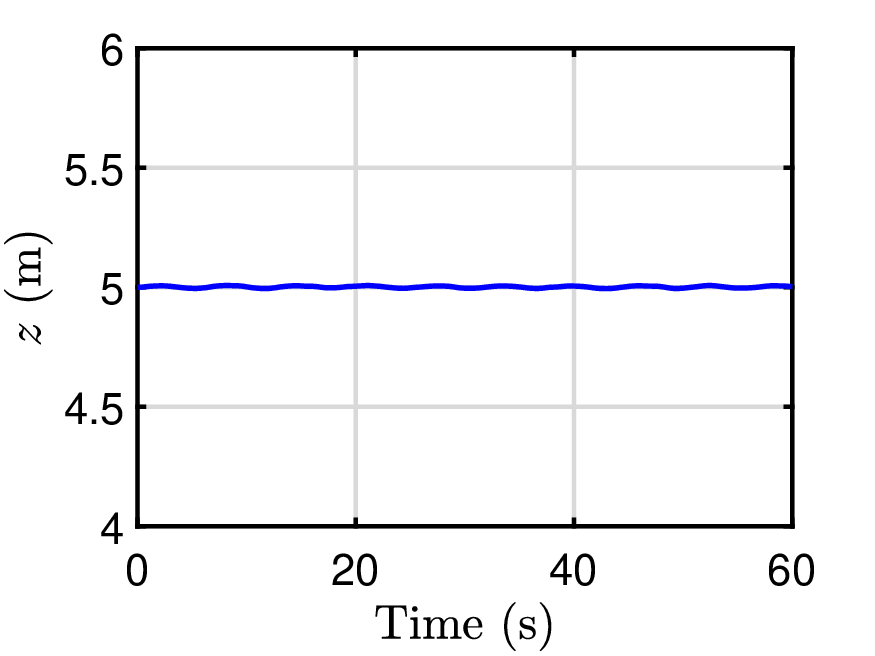}}
	\subfloat[{Roll angle } ]
	{\includegraphics[width=0.5\textwidth]{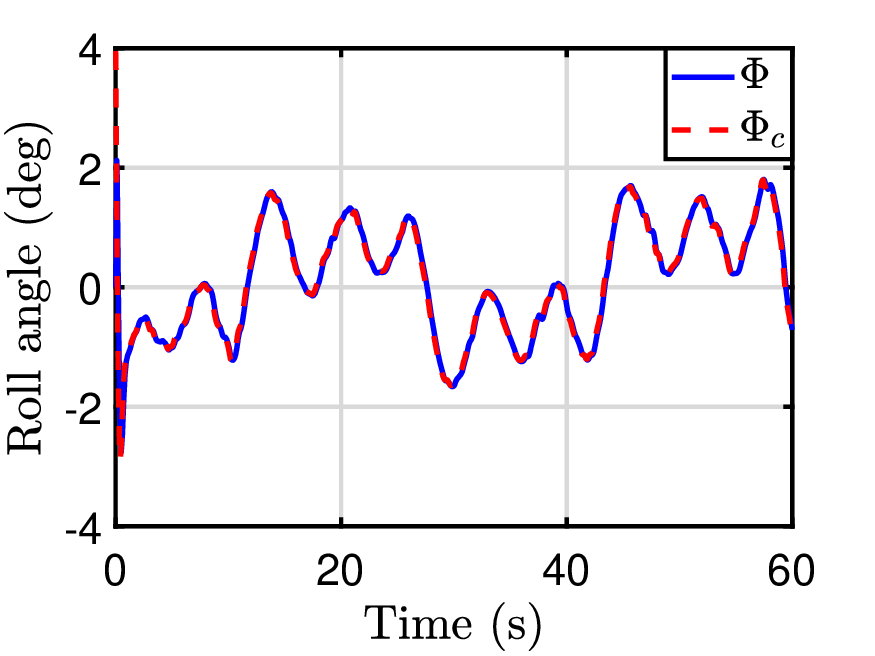}}\\
    \subfloat[{Pitch angle }]
	{\includegraphics[width=0.5\textwidth]{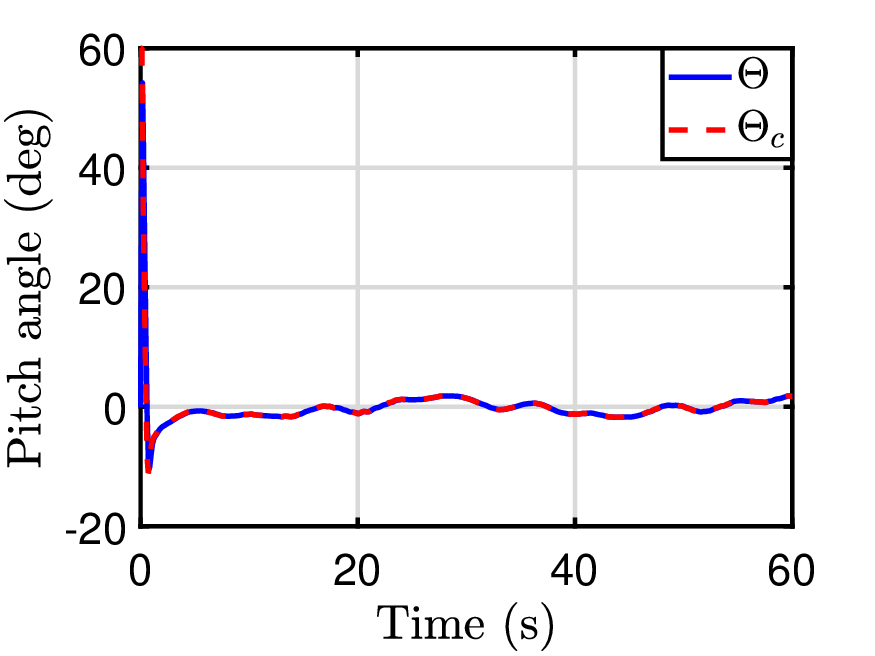}}
		\subfloat[{Velocity}]
	{\includegraphics[width=0.5\textwidth]{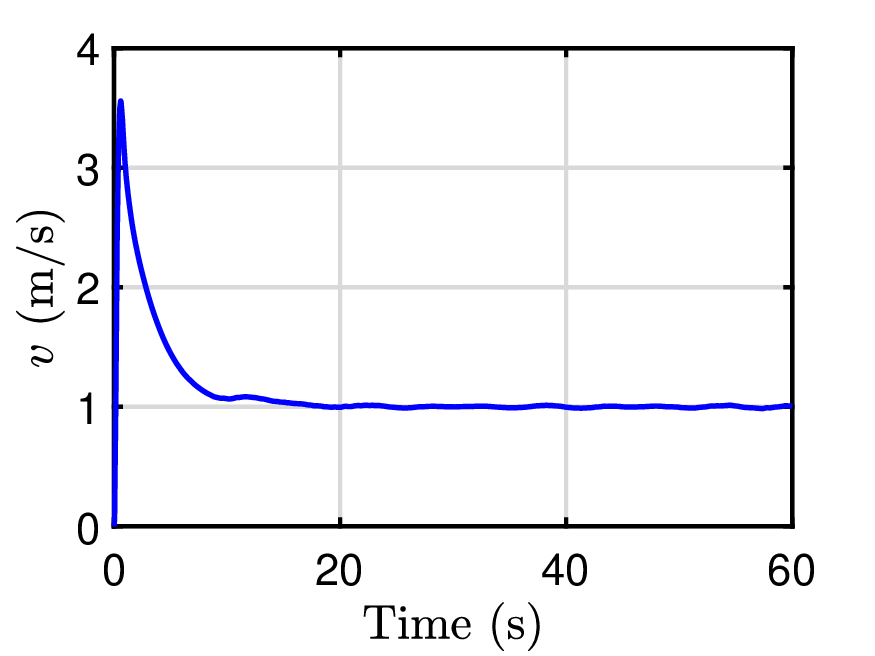}}
	\caption{Results for Case 3: Constant-speed elliptical circumnavigation.}
	 \label{case3}
\end{figure}

\begin{figure}
	\centering
	\subfloat[{ Quadrotor trajectories}]
	{\includegraphics[width=0.5\textwidth]{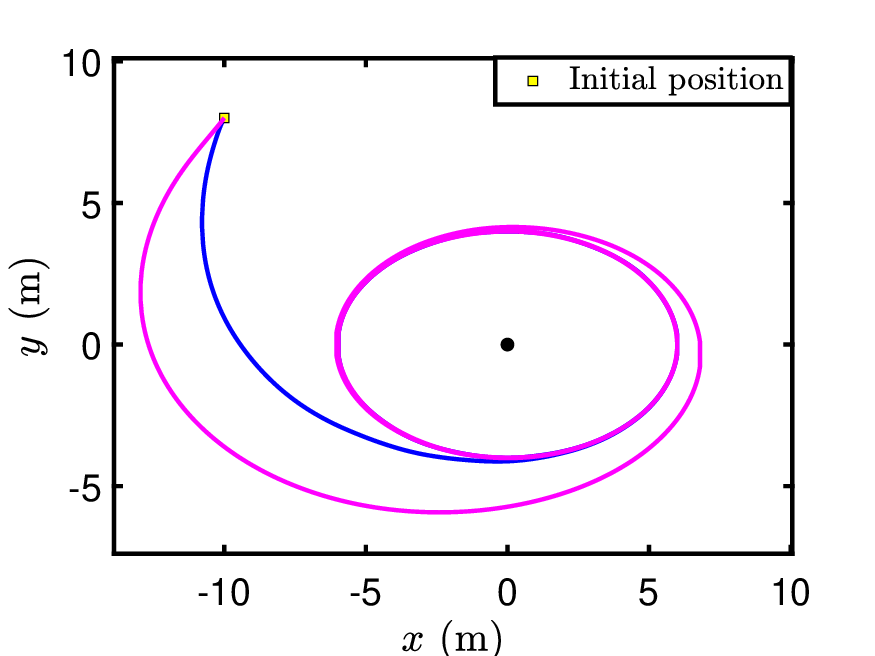}}
		\subfloat[{Elliptic radii}]
	{\includegraphics[width=0.5\textwidth]{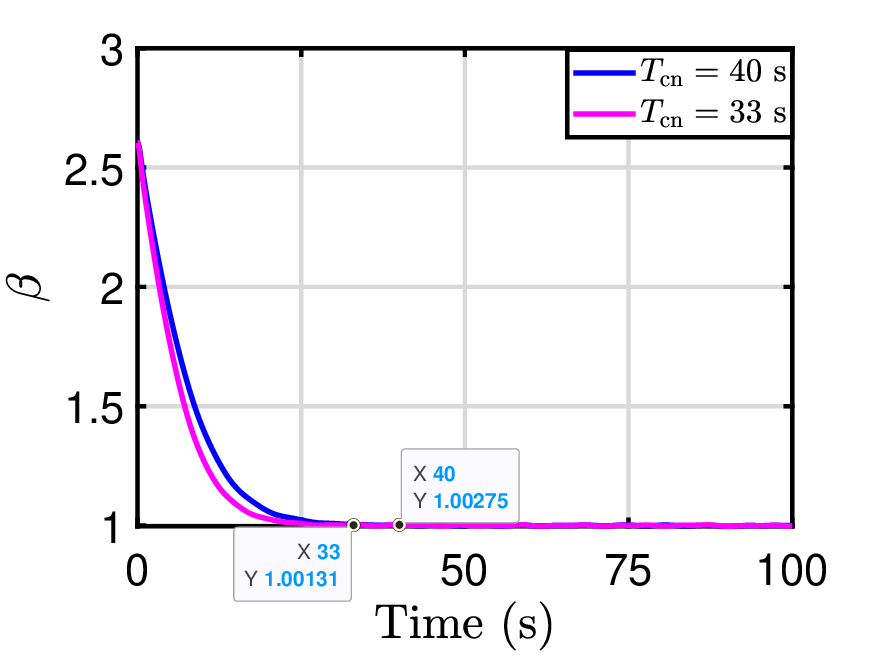}}\\
	\subfloat[{Elliptic angles}]
	{\includegraphics[width=0.5\textwidth]{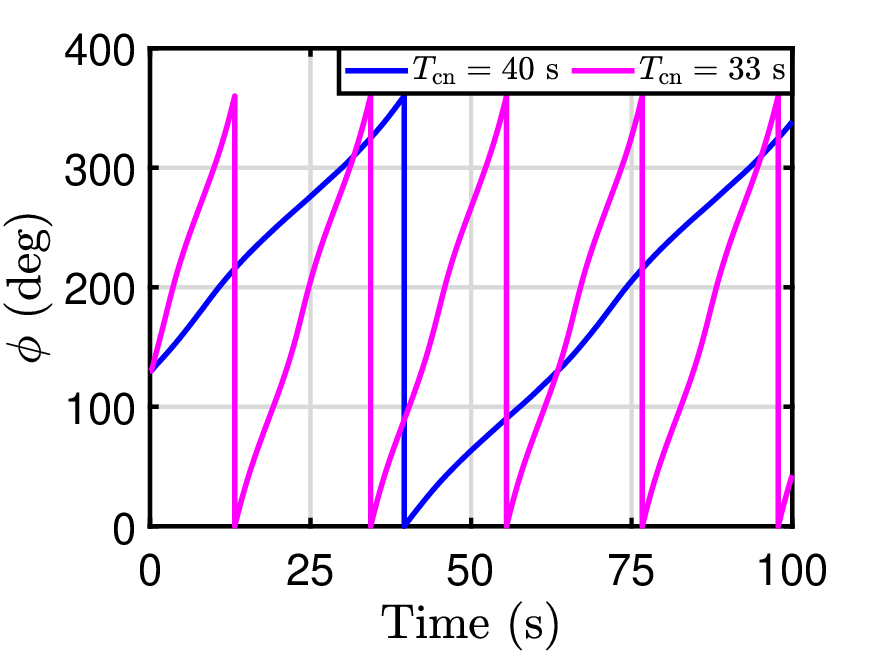}}
		\subfloat[{Velocities}]
	{\includegraphics[width=0.5\textwidth]{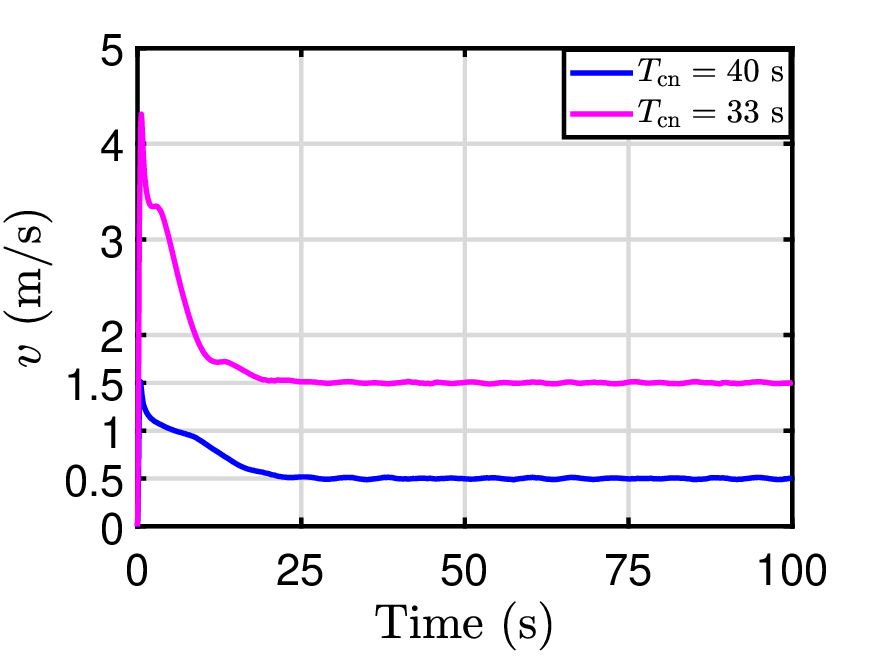}}
	\caption{Results for Case 3: Constant-speed elliptical circumnavigation for different settling times and desired speeds.}
	 \label{case3a}
\end{figure}

\subsection{Case 4: Dwell Time Constrained Elliptical Circumnavigation}
Consider the initial quadrotor position to be $(x_0, y_0) = (10~\mathrm{m}, 0~\mathrm{m})$ in this scenario, where the quadrotor is required to settle onto the boundary within $T_{cn} = 18.5~\mathrm{s}$. Subsequently, while traversing the boundary, it is desired to dwell for $T_{dw} = 10~\mathrm{s}$ around the viewpoint at $\theta_d = 130~\mathrm{deg}$ within an angular sector of $\delta = \pm 30~\mathrm{deg}$. Accordingly, the guidance parameters $\rho = 0.3679$ and $\lambda_{\beta} = 1.0070$ are determined using Eqs.~\eqref{p_parameter} and \eqref{lambda_beta_par}, respectively. The bifurcation parameter $\mu = 1.3166$ is computed with $\lambda_{\phi} = 0.2$ by numerically solving Eq.~\eqref{mu_dwell} using MATLAB's \texttt{fsolve} command. The resulting quadrotor trajectory is shown in Fig.~\ref{case4_traj}. The elliptic radial distance profile, shown in Fig.~\ref{case4}(a), indicates convergence to the  boundary within the specified settling time. The corresponding line-of-sight angle variation in Fig.~\ref{case4}(b) illustrates that the quadrotor dwells for $T_{dw} = 10~\mathrm{s}$ within the angular sector $[100, 160]~\mathrm{deg}$ with negligible deviation while maintaining the reference altitude, as depicted in Fig.~\ref{case4}(c). The roll and pitch responses are shown in Figs.~\ref{case4}(d) and \ref{case4}(e), respectively, while the corresponding velocity profile is presented in Fig.~\ref{case4}(f).

\begin{figure}
  \centering
  \includegraphics[width=0.52\textwidth]{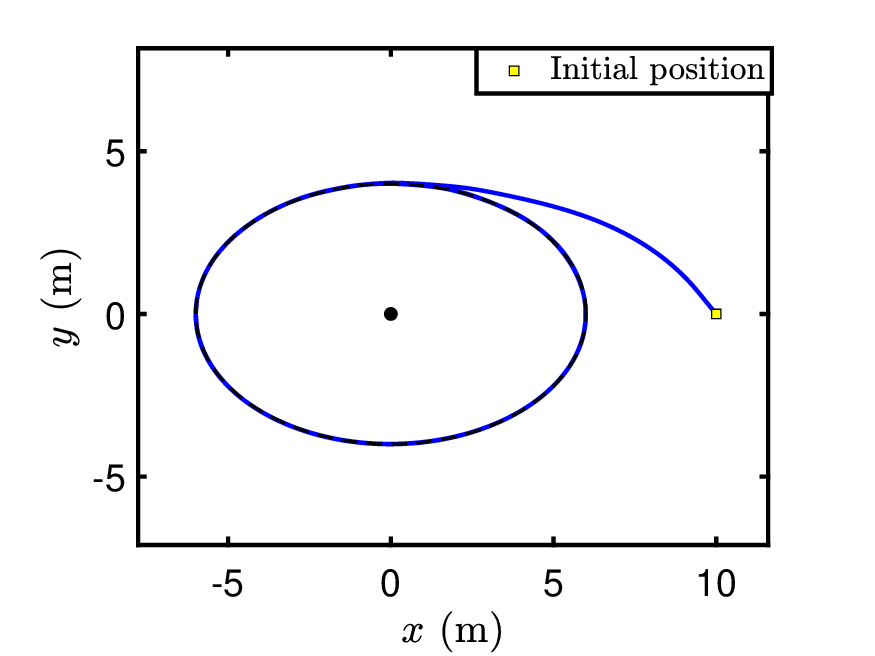}
  \caption{Quadrotor trajectory for Case 4.}
  \label{case4_traj}
\end{figure}

 \begin{figure}
	\centering
	\subfloat[{ Elliptic radius}]
	{\includegraphics[width=0.5\textwidth]{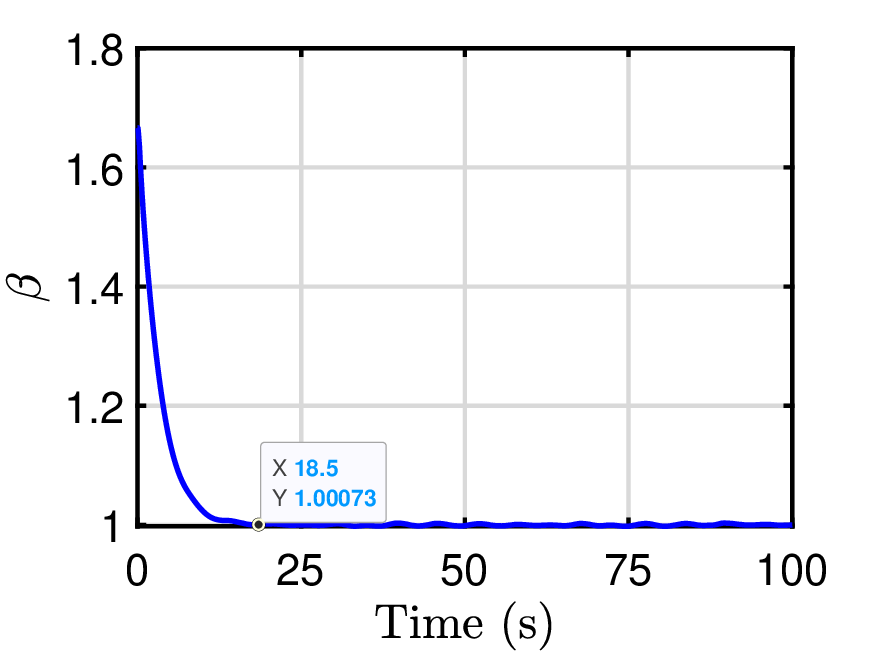}}
		\subfloat[{Line-of-sight angle }]
	{\includegraphics[width=0.5\textwidth]{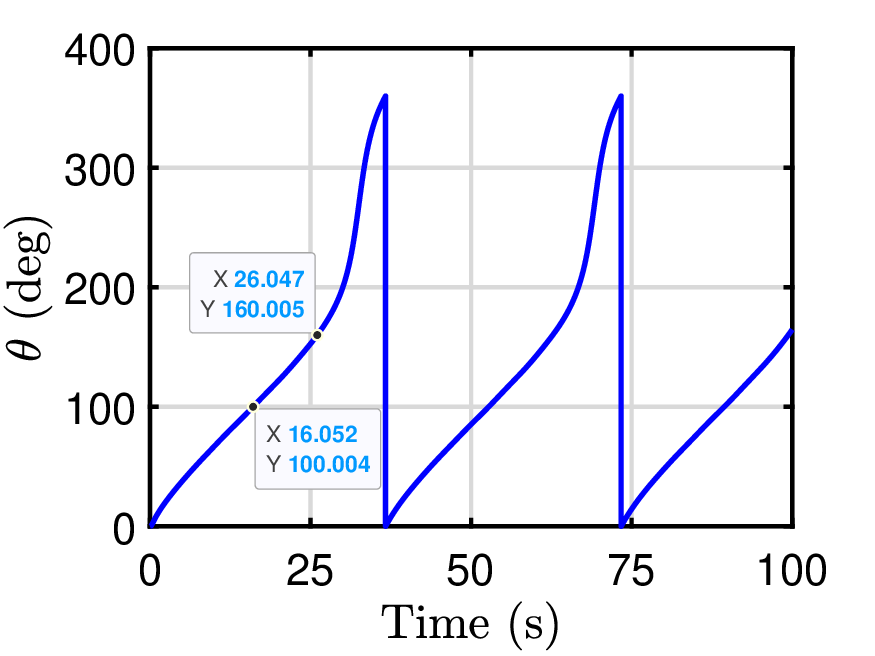}}\\
	\subfloat[{Altitude }]
	{\includegraphics[width=0.5\textwidth]{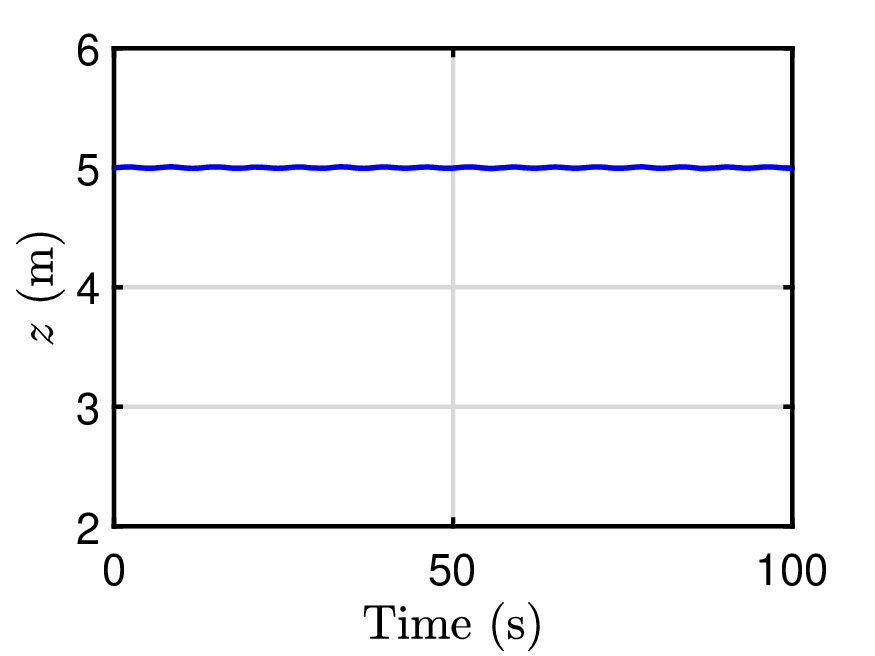}}
	\subfloat[{Roll angle} ]
	{\includegraphics[width=0.5\textwidth]{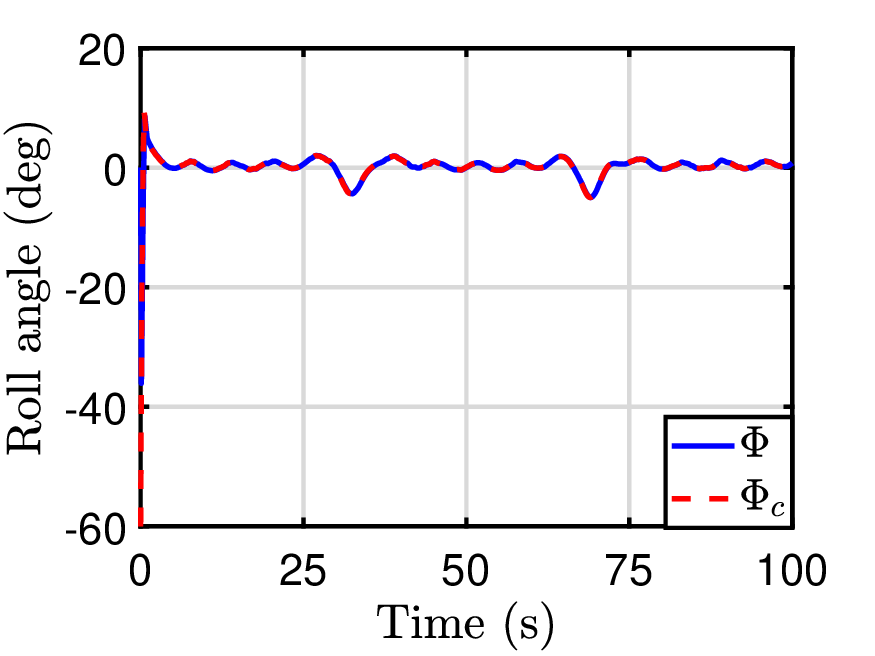}}\\
    \subfloat[{Pitch angle }]
	{\includegraphics[width=0.5\textwidth]{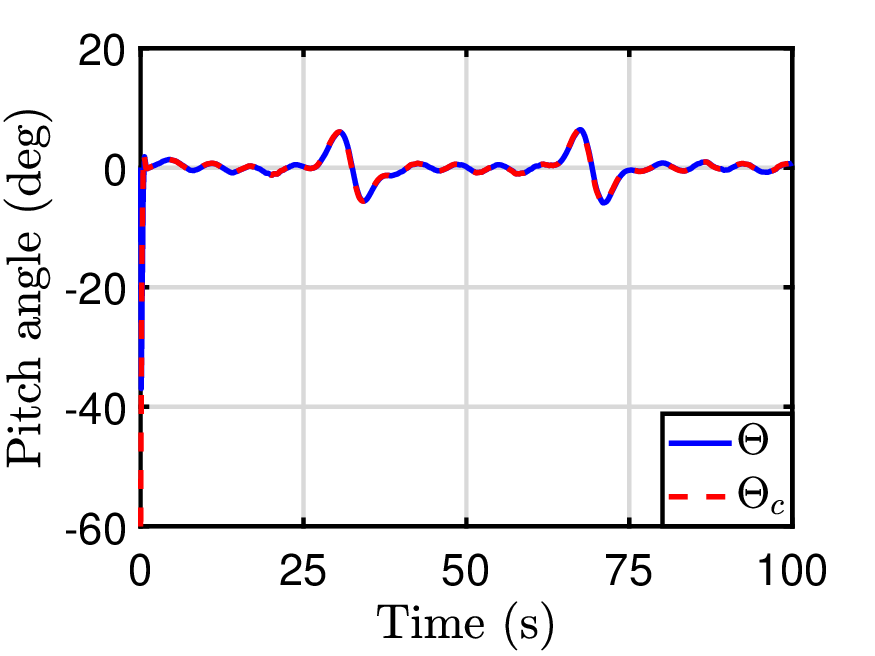}}
		\subfloat[{Velocity }]
	{\includegraphics[width=0.5\textwidth]{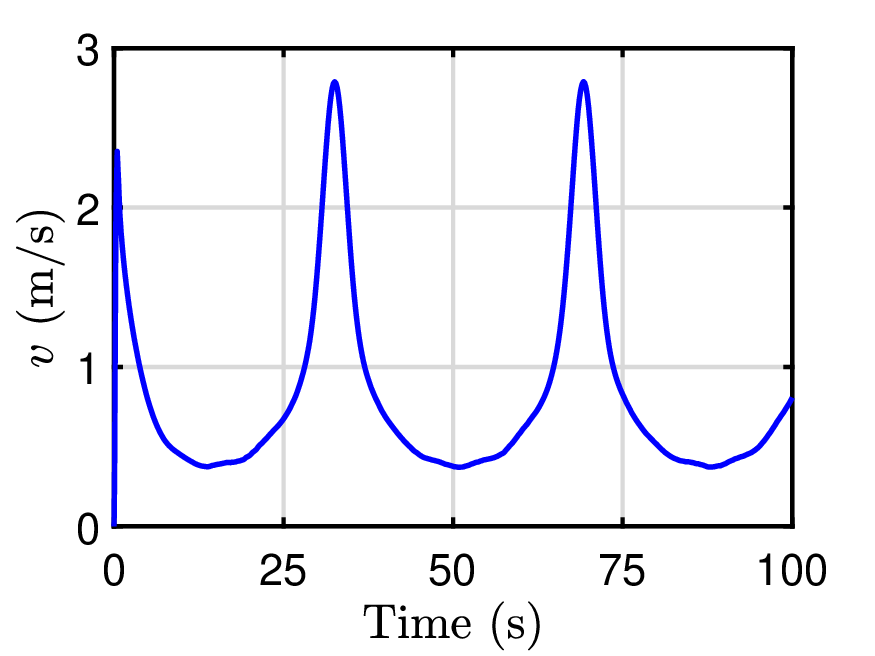}}
	\caption{Results for Case 4: Dwell-time constrained elliptical circumnavigation.}
	 \label{case4}
\end{figure}

\section{Conclusions}\label{sec:conclusions}
Using a bifurcation theory-based framework, this work proposes a novel quadrotor guidance method for carrying out distinct inspection tasks along an elliptical boundary. A new decoupled two-variable position representation defines the elliptical boundary. The proposed guidance method is developed using a saddle-node infinite-period bifurcation in the decoupled elliptic radial and angle dynamics, where the stable fixed point is matched to a desired viewpoint. The formulation ensures convergence to the desired viewpoint from every initial condition except those corresponding to the unstable fixed point of the elliptic angle dynamics.

The framework naturally supports additional tasks, such as switching between viewpoints and continuous elliptical circumnavigation. These tasks are accomplished by appropriately selecting the bifurcation and constant parameters in the elliptic-angle dynamics. Analytical expressions for settling time are also deduced using the radial control parameter and angular control coefficient in the elliptic radial and angle dynamics, respectively. Furthermore, an advanced circumnavigation feature is introduced that allows the quadrotor to dwell within a specified angular sector of the elliptical boundary for a prescribed duration. This task is achieved by exploiting stable limit cycles that arise naturally through variation of the bifurcation parameter, which can be numerically determined to satisfy the dwell-time constraint.

A qualitative comparison with existing elliptical circumnavigation approaches demonstrates the flexibility of the proposed method in enabling distinct trajectory modes through a unified formulation. The proposed guidance method yields closed-form and algebraically simpler velocity expressions for the quadrotor, facilitating practical implementation. The effectiveness of the proposed framework is demonstrated through extensive numerical simulations in which the reference trajectories generated by the proposed method are tracked using a 6-DOF quadrotor model in the presence of bounded disturbances, namely, model uncertainties and wind gusts modeled using the Dryden model. Overall, the results highlight the flexibility, robustness, and computational simplicity of the proposed method in accomplishing a diverse set of tasks using a unified guidance strategy. Future work will focus on extending the proposed method to three-dimensional geometrical boundaries, such as an elliptic cylinder, by incorporating elevation dynamics in addition to the elliptic radial and angular dynamics.

\section*{Appendix: Proof of Proposition~\ref{prop3}}
\begin{proof}
From Eq.~\eqref{desired_phi_dot}, the direction of approach of $\phi(t)$ toward a stable fixed point $\phi_d$ is determined by the sign of the initial elliptic angular rate $\dot{\phi}(0)$. Accordingly, following sets are defined when $ k \in \{-1, 1 \}$ as
\begin{equation*}
\mathbb{S}_{+} := \left\{ \phi_0 \,\big|\, \dot{\phi}(0) > 0 \right\}, \qquad
\mathbb{S}_{-} := \left\{ \phi_0 \,\big|\, \dot{\phi}(0) < 0 \right\}
\end{equation*}
Next, we consider two cases for $\mu$.
\subsection{$\mu \in (-1,1) \setminus \{0\}$}
We further consider two sub cases, that is, $\phi_0 \in \mathbb{S}_{+}$ and $\phi_0 \in \mathbb{S}_{-}$.
\subsubsection{$\phi_0 \in \mathbb{S}_{+}$}
From Definition~\ref{definition_stat}, the settling time $t_{s_{\phi}}$ satisfies
\begin{equation}
\phi(t \rightarrow \infty) - \phi(t_{s_{\phi}}) = \epsilon_{\phi}, \quad 0 < \epsilon_{\phi} \ll 1
\label{a1_revised}
\end{equation}
Using Eq.~\eqref{phi(t) expression} in Eq.~\eqref{a1_revised} yields
{\small
\begin{equation}
2 \tan^{-1}(m_2) - 2 \tan^{-1} \left(
\frac{ m_1 e^{- t_{s_{\phi}} \lambda_{\phi} \sqrt{k^2 - \mu^2}} - \frac{m_2 m_4}{m_3} }{ e^{- t_{s_{\phi}} \lambda_{\phi} \sqrt{k^2 - \mu^2}} - \frac{m_4}{m_3} }
\right) = \epsilon_{\phi}
\label{intermediate_1}
\end{equation} }
On simplifying Eq.~\eqref{intermediate_1}, leads to
{\small
\begin{equation}
t_{s_{\phi}} =
- \frac{1}{\lambda_{\phi} \sqrt{k^2 - \mu^2}} \ln \left(
\frac{ m_4 \left[ \tan \left( \tan^{-1}(m_2) - \frac{\epsilon_{\phi}}{2} \right) - m_2 \right] }{ m_3 \left[ \tan \left( \tan^{-1}(m_2) - \frac{\epsilon_{\phi}}{2} \right) - m_1 \right] }
\right)
\end{equation}}

\subsubsection{$\phi_0 \in \mathbb{S}_{-}$}
Similarly, from Definition~\ref{definition_stat},
\begin{equation}
\phi(t_{s_{\phi}}) - \phi(t \rightarrow \infty) = \epsilon_{\phi}
\end{equation}
Using Eq.~\eqref{phi(t) expression}, this leads to
{\small
\begin{equation}
t_{s_{\phi}} =
- \frac{1}{\lambda_{\phi} \sqrt{k^2 - \mu^2}} \ln \left(
\frac{ m_4 \left[ \tan \left( \tan^{-1}(m_2) + \frac{\epsilon_{\phi}}{2} \right) - m_2 \right] }{ m_3 \left[ \tan \left( \tan^{-1}(m_2) + \frac{\epsilon_{\phi}}{2} \right) - m_1 \right] }
\right)
\end{equation} }

\subsection{$\mu = 0$}
From \eqref{mu_parameter_p1} and \eqref{mu_parameter_m1}, the stable fixed point is $\phi^* = 0$ for $k = 1$ and $\phi^* = \pi$ for $k = -1$.

\subsubsection{$\phi_0 \in \mathbb{S}_{+}$}
Using Definition~\ref{definition_stat},
\begin{equation}
\phi^* - \phi(t_{s_{\phi}}) = \epsilon_{\phi}
\end{equation}
Using Eq.~\eqref{phi(t) expression}, the settling time is
\begin{equation}
t_{s_{\phi}} = - \frac{1}{\lambda_{\phi} k} \ln \left(
\frac{\tan\left( (\phi^* - \epsilon_\phi)/2 \right)}{\tan(\phi_0/2)}
\right)
\end{equation}

\subsubsection{$\phi_0 \in \mathbb{S}_{-}$}
In this case,
\begin{equation}
\phi(t_{s_{\phi}}) - \phi^* = \epsilon_{\phi}
\end{equation}
which leads to
\begin{equation}
t_{s_{\phi}} = - \frac{1}{\lambda_{\phi} k} \ln \left(
\frac{\tan\left( (\phi^* + \epsilon_\phi)/2 \right)}{\tan(\phi_0/2)}
\right)
\end{equation} 

\end{proof}

\section*{Acknowledgment}
The authors gratefully acknowledge the partial funding support provided by the Science and Engineering Research Board (SERB), Department of Science and Technology, Government of India, under Core Research Grant No. CRG/2022/008964. 

\bibliography{sample}

\end{document}